\documentclass[11pt,pdfa]{article}
  \usepackage[in]{fullpage}

\usepackage{iftex}
\ifPDFTeX
  \usepackage[utf8]{inputenc}
  \usepackage[noTeX]{mmap}
  \usepackage[T1]{fontenc}
\fi
\ifLuaTeX
  \usepackage{luatex85}
  \usepackage[noTeX]{mmap}
\fi

\usepackage[style=alphabetic,minalphanames=3,maxalphanames=4,maxnames=99,backref=true,backend=bibtex]{biblatex}

\usepackage[normalem]{ulem}
\usepackage{comment}
\usepackage[shortlabels]{enumitem}
\usepackage{bm}
\usepackage{amsfonts}
\usepackage{amsmath}
\usepackage{amssymb}
\usepackage{amsthm}
\usepackage{xcolor}
\usepackage{graphicx}
\usepackage{qtree}
\usepackage{tree-dvips}
\usepackage{float}
\usepackage{hyperref}
\usepackage[nameinlink]{cleveref}
\usepackage{braket}
\usepackage{mathrsfs}
\usepackage{tikz}
\usepackage{myqcircuit}

  \hypersetup{colorlinks={true},linkcolor={blue},citecolor=magenta}

  \theoremstyle{plain}
  \newtheorem{theorem}{Theorem}[section]
  \newtheorem{lemma}[theorem]{Lemma}
  \newtheorem{claim}{Claim}
  \newtheorem{corollary}[theorem]{Corollary}
  \newtheorem{proposition}{Proposition}
  \newtheorem{definition}[theorem]{Definition}

  \usepackage[style=alphabetic,minalphanames=3,maxalphanames=4,maxnames=99,backref=true]{biblatex}

\newcommand{\distr}{\mathcal{D}}
\newcommand{\rel}{{\cal R}}
\newcommand{\lang}{{\cal L}}

\newcommand{\verify}{\mathsf{Verify}}
\newcommand{\vrfr}{\verify}
\newcommand{\quant}{{\cal Q}}
\newcommand{\view}{\mathsf{View}}
\newcommand{\receiver}{\mathsf{Receiver}}
\newcommand{\accept}{\mathsf{Accept}}

\newcommand{\np}{\mathsf{NP}}
\newcommand{\protwi}{\Pi_{\sf wi}}
\newcommand{\relwi}{{\cal R}(\lang_{\sf wi})}
\newcommand{\open}{\mathsf{Open}}
\newcommand{\bfb}{{\bf b}}
\newcommand{\setup}{{\sf Setup}}
\newcommand{\sign}{{\sf Sign}}
\newcommand{\bfx}{{\bf x}}
\newcommand{\Sim}{{\sf Sim}}
\newcommand{\enc}{\mathsf{Enc}}
\newcommand{\functionspace}{f}

\newcommand{\poly}{\mathrm{poly}}

\newcommand{\hybrid}{\mathsf{H}}
\newcommand{\adversary}{\mathcal{A}}

\newcommand{\bfc}{\mathbf{c}}
\newcommand{\comm}{\mathsf{Comm}}

\newcommand{\cE}{\mathcal{E}}

\newcommand{\cH}{\mathcal{H}}

\newcommand{\prob}{\mathsf{Pr}}
\newcommand{\bfr}{\mathbf{r}}
\newcommand{\negl}{\mathsf{negl}}

\newcommand{\ketbra}[2]{\ket{#1}\!\bra{#2}}
\newcommand{\kb}[1]{\ket{#1}\!\bra{#1}}
\renewcommand{\cal}[1]{\mathcal{#1}}

\newcommand{\Tr}{\mathrm{Tr}}

\newcommand{\eps}{\varepsilon}

\newcommand{\calA}{\mathcal{A}}

\newcommand{\secparam}{\lambda}

\DeclareFontFamily{U}{skulls}{}
\DeclareFontShape{U}{skulls}{m}{n}{ <-> skull }{}

\title{Post-Quantum Zero-Knowledge with Space-Bounded Simulation\footnote{This work was done (in part) while the authors were visiting the Simons Institute for the Theory of Computing. This research was supported in part by the National Science Foundation under Grant No. NSF PHY-1748958.}}

\author{Prabhanjan Ananth\thanks{{\small Email: \texttt{prabhanjan@cs.ucsb.edu}}} \vspace{0.25em} \\ \small{UCSB} \and Alex B. Grilo\thanks{{\small Email: \texttt{alex.bredariol-grilo@lip6.fr}}} \vspace{0.25em} \\ {\small Sorbonne Universit\'e, CNRS, LIP6}}
\date{}

\begin{document}

\maketitle

\begin{abstract}
\noindent The traditional definition of quantum zero-knowledge stipulates that the knowledge gained by any quantum polynomial-time verifier in an interactive protocol can be simulated by a quantum polynomial-time algorithm. One drawback of this definition is that it allows the simulator to consume significantly more computational resources than the verifier. We argue that this drawback renders the existing notion of quantum zero-knowledge not viable for certain settings, especially when dealing with near-term quantum devices. 

In this work, we initiate a fine-grained notion of post-quantum zero-knowledge that is more compatible with near-term quantum devices. We introduce the notion of $(s,\functionspace)$ space-bounded quantum zero-knowledge. In this new notion, we require that an $s$-qubit malicious verifier can be simulated by a quantum polynomial-time algorithm that uses at most $\functionspace(s)$-qubits, for some function $\functionspace(\cdot)$, and no restriction on the amount of the classical memory consumed by either the verifier or the simulator. 
\par We explore this notion and establish both positive and negative results: 
\begin{itemize}
    \item For verifiers with logarithmic quantum space $s$ and (arbitrary) polynomial classical space, we show that $(s,\functionspace)$-space-bounded QZK, for $\functionspace(s)=2s$, can be achieved based on the existence of post-quantum one-way functions. Moreover, our protocol runs in constant rounds. 
    \item For verifiers with super-logarithmic quantum space $s$, assuming the existence of post-quantum secure one-way functions, we show that  $(s,\functionspace)$-space-bounded QZK protocols, with fully black box simulation (classical analogue of black-box simulation)  can only be achieved for languages in BQP. 
\end{itemize}
\end{abstract}

\section{Introduction}

Zero-knowledge is a foundational notion in cryptography. Invented in the 80s by Goldwasser, Micali and Rackoff~\cite{GMR}, this notion has slowly transitioned from being a purely theoretical notion to having applications in practice. Some notable applications of zero-knowledge include secure computation~\cite{GMW} and cryptocurrencies~\cite{zcash}.
\par A zero-knowledge proof (or argument) system for a language in NP is an interactive protocol between a prover, who receives as input an instance-witness pair $(x,w)$ and a verifier, who receives as input an instance $x$. The zero-knowledge property states that the verifier, after interacting with the prover, should not gain any information about the witness beyond what is leaked by the statement. That is, in the words of the inventors of zero-knowledge~\cite{GMR}, {\em ``$\ldots$ what the verifier sees in the protocol (even if he cheats) should be something which the verifier could have computed himself$\ldots$"}. Formally, we define the notion of a simulator and we require that the view of the verifier after interacting with the prover should be computationally indistinguishable from the output of the simulator. 
\par  When we consider polynomial-time verifiers we typically require the simulator to run in polynomial time. In particular, the storage and the computational power utilized by the simulator can be an arbitrary polynomial, respectively, in the storage and computational power of the verifier. One might worry that this definition is not strong enough: the verifier might gain knowledge from a protocol that would need far more computational resources in order for it to have computed by itself.\footnote{See the discussion on precise zero-knowledge in \Cref{sec:related-work}.} Nonetheless, this is a reasonable assumption to make in the present day as getting access to more computational resources is cheaper than ever before.  

\paragraph{Quantum Zero-Knowledge.} The advancement in quantum technology forces us to re-think the design of cryptographic protocols. But even before we build protocols that are secure against quantum adversaries, we need to first focus on formal definitions and appropriately define post-quantum security for existing cryptographic primitives. We focus on the notion of zero-knowledge against malicious quantum verifiers, also referred to as quantum zero-knowledge. One approach to formalize this definition is to modify the definition of classical zero-knowledge as follows: instead of considering (classical) probabilistic polynomial-time adversaries in the traditional zero-knowledge definition, we instead consider quantum polynomial-time (QPT) adversaries. And instead of modeling the auxiliary input to the verifier as a classical string, we model it as a quantum state. Almost all the works~\cite{Wat09,BS20,AL20,ABGKM20,ACP21,BKS21,CCY21,LMS21} take this approach of defining quantum zero-knowledge. While this is the most general definition one can think of, this definition does not accommodate the subtleties on the computational power of the quantum verifier that is highlighted by the example below.

\paragraph{Definitional Issues.} Let us consider the following language based on integer factorization:
$$ \lang_{\sf FACTORING} = \{ \langle N, L, U \rangle : \exists \text{ prime } p \text{ that divides } N, \text{ s.t.} L \leq p \leq U\}.$$

\noindent Consider the following protocol for $ \lang_{\sf FACTORING}$ between a prover $P$ and a verifier $V$:
\begin{enumerate}
    \item $P$ sends a prime $p$ that divides $N$  and moreover, $L \leq p \leq U$
    \item $V$ checks if: $a)$ $p$ is prime; $b)$ $p$ divides $N$; and $c)$ $L \leq p \leq U$. If all the checks pass, then $V$ accepts. Otherwise, $V$ rejects. 
\end{enumerate}

\noindent It is not hard to see that this protocol is complete and sound. Interestingly, according to the existing definition, this protocol satisfies the quantum zero-knowledge property. The simulator works as follows:
\begin{enumerate}
    \item Using the Shor's algorithm for integer factorization~\cite{Shor99}, find all the prime factors $p_1,...,p_\ell$ of $N$.
    \item Let $i$ be such that $p_i$ is prime, $p_i$ divides $N$ and $L \leq p_i \leq U$. Send $p_i$ to the verifier.
\end{enumerate}

\noindent Notice that when describing the simulator, we did not take into consideration the computational power of the verifier. In particular, let us consider a malicious verifier $V'$ which is a hybrid computer, consisting of a quantum device available today (such as the implementations of low-depth Random Circuit Sampling devices~\cite{Google}, Boson sampling~\cite{boson}, non-universal adiabatic computation~\cite{DWave}, etc.) and a classical machine. None of the currently available quantum devices have the capability to factor large prime numbers. If $V'$ participates in the protocol described above then $V'$ could be gaining knowledge, i.e., a non-trivial factor of $N$, that it could not have been able to compute by itself; despite the protocol being quantum zero-knowledge! 

This discrepancy appears because the definition of quantum zero-knowledge allows the simulator to be an arbitrary quantum polynomial-time algorithm, regardless of the computational power of the verifier. For instance, even if the malicious verifier is a classical polynomial-time algorithm, the simulator is still a quantum polynomial-time algorithm. 
\par A more realistic definition should instead consider the resources used by the verifier and require the simulator to run under (roughly) the same constraints. There are many resources we can take into consideration when formulating this definition. In this work, we focus on the resource of quantum memory. More specifically, we focus on the number of qubits utilized by the simulator in relation to the number of qubits used by the verifier.   

\paragraph{Constructions of QZK: Prior Work.} We first observe that the existing works on post-quantum zero-knowledge propose simulators whose space complexity is a large polynomial overhead in the space complexity of the verifier. One main reason behind this is the fact that all the existing black-box simulation techniques~\cite{Wat09,Unruh12,CCY21,CMSZ22,CCLY21,LMS21} first purify the verifier; that is, they convert the next message function of the verifier into a unitary\footnote{Technically, it is converted into a unitary followed by measurement, where the measurement outcome will be the message communicated to the prover.} by delaying the measurements. The purification process increases the space complexity, proportional to the number of intermediate measurements. %
\par Even if we ignore the above issue and consider verifiers where the next message functions are implemented as unitaries, the existing simulators still have large polynomial overheads in the space complexity of the verifier. In~\cite{Wat09}, the reason for this is the fact that the space complexity of the simulator depends on the communication complexity of the protocol which in turn is some function of the security parameter. In~\cite{CCY21,CCLY21,LMS21}, the simulator runs the verifier coherently multiple times and thus, the space complexity additionally depends on the number of iterations. 

\subsection{Our Contributions}
We propose a new definition of quantum zero-knowledge, where the amount of quantum memory used by the simulator is closely related to the quantum memory used by the verifier. We also investigate the feasibility of the new notion. 

\paragraph{1. New Definition: Space-Bounded QZK.} We formulate a new definition of post-quantum zero-knowledge, that we call {\em $(s,\functionspace)$-space-bounded} QZK, where $s \in \mathbb{N}$ and $\functionspace(\cdot)$ is some function. Suppose the malicious quantum polynomial-time (QPT) verifier is a quantum algorithm that uses at most $s$ qubits of quantum memory. Then, we require that the simulator runs in polynomial time and use the number of qubits is at most $\functionspace(s)$. For example, if $\functionspace$ is defined to be $\functionspace(x)=2x$ then $(s,\functionspace)$-space bounded QZK states that the simulator should take up at most twice the number of qubits as the verifier. On the other hand, we don't place any restriction on the classical memory of the simulator. The amount of classical storage of the simulator could be an arbitrary polynomial in the amount of classical storage of the verifier. Indeed, in reality, classical storage is much cheaper than quantum storage and thus, we need to be more precise about modeling the latter. 
\par We consider three different notions of $(s,\functionspace)$-space-bounded QZK:
\begin{itemize}
    \item {\em Fully Black Box}: In this notion, the simulator gets oracle access to the verifier. More precisely, each round of the verifier is modeled as a sequence of channels (one per round) and the simulator can make polynomially many queries to each of these channels. This notion of quantum zero-knowledge is a direct analogue of classical zero-knowledge. 
    
    In terms of space complexity, the space undertaken by the simulator would also take into account the space needed to store the private state of the verifier in-between the executions of every two consecutive rounds. 
    
    In this work, we mainly focus on understanding the feasibility of $(s,\functionspace)$-space bounded QZK for different $s(\cdot)$ functions. 
    \item {\em Black Box}: In this notion, the simulator gets oracle access to the purification of the verifier and its inverse. More precisely, suppose the verifier is represented as a sequence of channels and let $U_1,\ldots,U_k$ be their {\em canonical} purifications\footnote{This means that the purification is computed in a specific manner: that is, by delaying the measurements of the channel.}. Then, the simulator gets oracle access to $U_i$ and also oracle access to $U_i^{\dagger}$, for every $i \in [k]$. Although, in this model, the verifier's code is modified before giving access to the simulator, we still call it black box in order to be consistent with most of the previous works on black box quantum zero-knowledge who adopt this model. 
    
    Note that here, we model the quantum space complexity of the simulator as a function of the quantum space complexity of the original verifier (and not the one obtained after purification). A (canonically) purified verifier can take significantly more space than the underlying non-purified verifier. For example, purifying an  $s$-qubit verifier with $\ell$ measurements will result in a unitary that consumes at least $s+\ell$ qubits. Thus, $(s,\functionspace)$ space-bounded black-box QZK, for any polynomial $\functionspace(\cdot)$, is impossible to achieve. On the other hand, prior works~\cite{Wat09,ACP21,CCY21,CCLY21,LMS21} are $(s,\functionspace)$-space bounded black box QZK for super-polynomial $\functionspace(s)=2^{\omega(\log(s))}$\footnote{Notice that in these previous results, the number of qubits used by the simulator is polynomial but it might scale with the {\em time-complexity} of the verifier and that is why we can only achieve a super-polynomial upper-bound on the number of qubits.}.
    \item {\em Non Black Box}: Finally, one can consider non black box quantum zero knowledge, where the simulator is allowed to arbitrarily depend on the verifier and in particular, could make use of the code of the verifier. This definition resembles the counterpart definition of (classical) non-black box zero-knowledge. Prior works~\cite{BS20,BKS21} satisfy $(s,\functionspace)$-space bounded non black box QZK for super-polynomial $\functionspace(s)=2^{\omega(log(s))}$. We leave the investigation on the feasibility of $(s,\functionspace)$-space bounded non black box QZK, when $\functionspace(\cdot)$ is a {\em polynomial}, as an interesting open problem.
\end{itemize}

\paragraph{2. Space-Bounded QZK Against Logarithmic Space Verifiers.} We first focus on a restricted case case when the malicious QPT verifiers only have access to logarithmically many qubits. In this case, we demonstrate a feasibility result. 
\par Suppose $s=O(\log(\secparam))$.  We show that there exist, even constant-round, $(s,\functionspace)$-space-bounded quantum zero-knowledge protocols for NP, where $\functionspace(x)=2x$, based on the assumption of post-quantum one-way functions. In fact, our protocol even achieves fully black-box quantum zero-knowledge. 

\paragraph{3. Space-Bounded QZK Against Super-Logarithmic Space  Verifiers.} We then investigate the case when the malicious verifiers have access to super-logarithmically many qubits. In this case, we present a negative result. 
\par Suppose $s=\omega(\log(\secparam))$. Assuming the existence of post-quantum one-way functions, we show that there do not exist fully black-box $(s,\functionspace)$-space bounded quantum zero-knowledge protocols for languages outside BQP. 
\par Our negative result only applies to fully black-box quantum zero-knowledge and it is an interesting open problem to either prove a negative result or circumvent our negative result using non black box techniques.

\subsection{Technical Overview}
\noindent In~\Cref{sec:overview:firstresult} and~\Cref{sec:overview:secondresult}, we present an overview of the techniques employed in both the results.

\subsubsection{Zero-Knowledge Against Logarithmic Quantum Space Verifiers} 
\label{sec:overview:firstresult}
Rewinding has been the quintessential technique employed in proving black-box simulation security of cryptographic protocols against probabilistic polynomial-time classical adversaries. A rewinding-based simulator has to store copies of the intermediate states of the verifier so that it can use these copies whenever it rewinds the execution of the verifier to an earlier round. However, adopting rewinding to prove post-quantum security has not been easy; this was first identified by~\cite{vanGraaf97}. At a high level, the reason for the difficulty arises since rewinding implicitly requires the ability to store copies of the intermediate states, which the no-cloning theorem~\cite{Dieks82,WZ82}.
Thus, most of the recent works on post-quantum zero-knowledge~\cite{Wat09,Unruh13,ACP21,CCY21,CMSZ22,CCLY22,LMS21} have proposed various rewinding techniques to perform simulation without having the need to store intermediate copies of the verifier's state. 

\paragraph{State Recovery.}\ To rewind the verifier to an earlier round, the most commonly adopted strategy is to invert the operations performed by the verifier in the hope that we can completely recover the earlier states. There are two issues with it:
\begin{enumerate}
    \item Firstly, this means that the verifier has to be a unitary and thus needs to be {\em purified}; that is, all the intermediate measurements of the verifier needs to be pushed to the end of each round. As a consequence of the purification process, the simulator could take much larger quantum space than the verifier. 
    \item Secondly, each round could potentially disturb the verifier's intermediate state in an irreversible way. Thus, we might no longer be able to recover the earlier states. 
\end{enumerate}
All the recent works on black-box quantum zero-knowledge give up on the first issue. Regarding the second issue, they provide a workaround by showing that still retains some useful properties about the original intermediate state, even if the recovered state is quite different from the original intermediate state. For example, the probability that the verifier does not abort on the recovered state could be at least as much as the probability that the verifier does not abort on the original intermediate state. We will not go into the details of how state recovery is performed in each of the recent works~\cite{Wat09,CMSZ22,CCY21,LMS21,CCLY21} since it is not relevant to our approach.

\paragraph{An observation.}\ Consider an $M$-qubit quantum state $\rho$, where $M = O(\log(\secparam))$. Our observation is that a maximally mixed state $\frac{I}{2^M}$ is a good approximation of $\rho$. That is, suppose there exists a binary positive operator valued measure (POVM) $\Lambda$ such that $p = \Tr(\Lambda(\rho))$ then it holds that $p = \Tr(\Lambda(\frac{I}{2^M})) \geq \frac{p}{2^M}$. This suggests a new approach to perform rewinding: to rewind the verifier to an earlier round, say $i$, simply execute the $i^{th}$ round of the verifier on input $\frac{I}{2^M}$. The advantage of this approach is two-fold. Firstly, there is no need to purify the verifier! %
Secondly, we only need $M$ additional qubits of quantum storage to initialize the maximally mixed state. 

\paragraph{Case Study: Four Round Protocol due to Goldreich-Kahan~\cite{GK90}.} Let us try to use this idea to prove post-quantum security of a zero-knowledge protocol. We will start with the four-round protocol due to Goldreich and Kahan~\cite{GK90}. Denote the four rounds in the protocol to be $(\alpha,\beta,\gamma,\delta)$. \\

\noindent {\em Simulating Classical Verifiers}: We first recall the strategy to simulate classical adversaries. The classical simulator, running in expected polynomial time, receives as input $\alpha$, which is nothing but the commitment of verifier's challenges. Let the verifier's state at this point be ${\sf st}$. 
\begin{itemize}
    \item {\em First Step. Execution of Main Thread}: The simulator then sends $\beta$, which is nothing but the commitment of 0. There are two things that can happen. Either the verifier aborts, in which case, even the simulator aborts. Or it could happen that the verifier does not abort. In this case, the simulator moves on to the second step. 
    \item {\em Second Step. Execution of Lookahead Threads}: At this point, the simulator has the opening of $\alpha$, denoted by $\gamma$. The simulator then rewinds the verifier to ${\sf st}$, i.e., the state where it just sent $\alpha$. It uses the information from the opening $\gamma$, to compute the commitment $\beta^*$. It sends $\beta^*$ to the verifier. If the verifier sends $\gamma^*$ then the simulator completes the transcript $(\alpha,\beta^*,\gamma^*,\delta^*)$, and additionally outputs the verifier's final state. If the verifier aborts then the simulator keeps repeating the rewinding process until it wins.  
\end{itemize}
\noindent The important thing to remember is that the simulator first runs the main thread to decide whether to continue or not. If it continues, it outputs one of the lookahead threads that did not abort. \\

\noindent {\em Simulating Quantum Verifiers}: Suppose we need to simulate quantum verifiers. We extend the above approach to the quantum setting as follows. The quantum simulator receives as input $\alpha$, which is nothing but the commitment of verifier's challenges. Let the verifier's state at this point be $(\rho,{\sf st})$, where $\rho$ is a $M$-qubit quantum state and ${\sf st}$ is a classical string. 
\begin{itemize}
    \item The first step is the same as above. The simulator executes the verifier on input the state $(\rho,{\sf st})$. Let the resulting state of the verifier be $(\sigma,{\sf st}')$. 
    \item If the verifier has not aborted at this point, it moves to the second step. At this point, if the simulator has to rewind, it no longer has a copy of $\rho$ since it was computed upon in the first round. Thus, the simulator uses the maximally mixed state $\frac{I}{2^M}$ in place of $\rho$ whenever it rewinds.
\end{itemize}
As per our observation earlier, $\frac{I}{2^M}$ serves as a good approximation of $\rho$. In particular, if $p$ is the probability that the verifier on input $\rho$ outputs a valid third message then the verifier on input $\frac{I}{2^M}$ outputs a valid third message with probability at least $\frac{p}{2^M}$. It can be argued that, as long as $M$ is logarithmic, the simulator in expected polynomial time can recover a valid transcript $(\alpha,\beta^*,\gamma^*,\delta^*)$ from one of the lookahead threads just like the classical simulator. However, this is not sufficient. The quantum simulator should not also output the residual state of the verifier along with the transcript $(\alpha,\beta^*,\gamma^*,\delta^*)$. But note that the residual state of the verifier obtained in the lookahead thread is useless: it is the state obtained by replacing the intermediate state $\rho$ with a maximally mixed state. One option is to output $(\sigma,{\sf st}')$, which is the state obtained in the main thread, along with $(\alpha,\beta^*,\gamma^*,\delta^*)$. However, the state $(\sigma,{\sf st}')$ is obtained in the main thread and hence, is inconsistent with the lookahead thread transcript $(\alpha,\beta^*,\gamma^*,\delta^*)$. Thus, we need a different approach where we can simulate in such a way that the joint distribution of the transcript and the verifier's state is computationally indistinguishable from the real world. 

\paragraph{Protocol Template.} We propose a different protocol and prove their post-quantum security using our key observation. As a starting point, we consider a commitment scheme. For the current discussion, we will consider commitments that are non-interactive and satisfy perfect binding and computationally hiding property. In the technical sections, we relax the requirement to allow the commitment to be interactive and moreover, we only require statistical binding.
\par We present a simplified version of the protocol template first. Let $\lang$ be the NP language associated with the protocol. We denote the prover to be $P$ and the verifier to be $V$. The prover receives as input instance-witness pair $(x,w)$ and the verifier receives as input $x$.  
\begin{enumerate}
    \item $V$ sends two commitments of $\alpha_0$ and $\alpha_1$.  
    \item $P$ sends a bit $b$. 
    \item $V$ sends the opening of commitment of $\alpha_b$. 
    \item $P$ uses a witness-indistinguishable protocol to prove to $V$ that either $x \in \lang$ or it knows $\alpha_0 \oplus \alpha_1$.  
\end{enumerate}
\noindent The above protocol is sound since a malicious prover can only receive one of the two decommitments and thus, will not know $\alpha_0 \oplus \alpha_1$ in order to complete the WI phase. To prove that the above template satisfies post-quantum zero-knowledge, we first need to demonstrate a simulator.  
\begin{itemize}
    \item The simulator receives as input commitments of $\alpha_0$ and $\alpha_1$ from the malicious verifier. Let $
    (\rho,{\sf st})$ be the state of the verifier at this point, where $\rho$ is the $M$-qubit state and ${\sf st}$ is the classical state. 
    \item {\em First Step. Execution of Main Thread}: The simulator sends a random bit $b$. It receives an opening of the commitment of $\alpha_b$ from the verifier. Let the state of the verifier at this point to be $(\sigma,{\sf st}')$.  
    \item {\em Second Step. Execution of Lookahead Threads}: The simulator then rewinds the verifier to just after the first round. This means that the simulator restarts the verifier on the state $\left(\frac{I}{2^M},{\sf st} \right)$. The simulators sends a random bit $b'$ to the verifier and it hopes that $b' \neq b$ and the verifier does not abort. The simulator keeps repeating this process until it succeeds.    
\end{itemize}
\noindent Once the simulator succeeds, it has both $\alpha_0,\alpha_1$. It discards all the lookahead transcripts and only retains the main thread transcript. Continue the execution of the protocol with the verifier, where the state of the verifier is $(\sigma,{\sf st}')$. The simulator then convinces the verifier in the WI phase using $\alpha_0 \oplus \alpha_1$. 
\par Note that unlike Goldreich-Kahan, the simulator retains the main thread transcript along with the verifier's residual state, just after the rewinding process, whose joint distribution is computationally indistinguishable from the real world.
\par While the main approach is sound, implementing the above approach encounters some issues. We mention the issues and how we fix them. The first issue is that it could happen that the verifier doesn't abort with non-negligible probability and also, that the simulator never succeeds when the verifier does not abort. Indeed, the verifier could decide that it will not abort if and only if the bit $b$ sent by the prover is 0. In this case, the simulator will never be able to recover two valid transcripts if the verifier does not abort. We fix this by increasing the challenge space. Instead of requiring the verifier to send two commitments, it sends $2 \secparam$ commitments of messages $((\alpha_{1,0},\alpha_{1,1}),\ldots,(\alpha_{\secparam,0},\alpha_{\secparam,1}))$ and in the third message, it opens commitments of messages $(\alpha_{1,b_1},\ldots,\alpha_{\secparam,b_\secparam})$, where $b_1,\ldots,b_\secparam$ are the challenge bits sent by the prover in the second message. If the verifier does not abort with non-negligible probability then there exists an index $i$ such that the verifier opens to each value $\alpha_{i,0},\alpha_{i,1}$ with non-negligible probability. The second issue is on non-malleability. The prover could launch a malleability attack by leveraging the commitments sent by the verifier in the first message to cheat the verifier in the WI phase. A first attempt is to make the prover commit to the witness to be used in the WI phase. But even this is susceptible to the malleability attack. We use a technique explored in a recent work~\cite{AL20}, where the prover commits to some randomness ${\bf r}^*$ at the very beginning even before it sees any message from the verifier. At a later point in time, when the prover commits to the witness, it is expected to use this randomness ${\bf r}^*$. This prevents the mauling attack since if the prover computed its own commitment by mauling the commitment by the verifier, it would not have the knowledge of the randomness contained in its commitment. Once we incorporate these fixes, we have a complete description of our protocol. 

\subsubsection{Zero-Knowledge Against Super-Logarithmic Quantum Space Verifiers} 
\label{sec:overview:secondresult}

 In~\Cref{sec:overview:firstresult}, we restricted our attention to zero-knowledge against logarithmic quantum space verifiers.  We now justify the restriction on the quantum space of the verifiers by showing the following: assuming post-quantum one-way functions, only for languages in BQP, there exist protocols that are space-bounded fully black-box QZK against super-logarithmic quantum space verifiers. 

This is shown by first proving that any fully black-box QZK protocol against super-logarithmic quantum space verifiers should have a straight-line simulator. We design a contrived verifier $V'$ that is composed of channels $\Phi_1,\ldots,\Phi_k$, where $k$ is the number of rounds in the protocol and $\Phi_i$ is used to compute the $i^{th}$ round query and suppose $M=\omega(\log(\secparam))$. We argue that any simulator simulating this verifier should be of a specific form: it should first make a successful query to $\Phi_1$, then make a successful query to $\Phi_2$ and so on. For any $i$, a query made by the simulator to $\Phi_i$ is deemed successful\footnote{In the technical sections, we call such queries non-abort queries (\Cref{def:nonabort:queries}).} if $\Phi_i$ does not abort on that query. Note that a straightline simulator could make many unsuccessful queries to any $\Phi_j$ (i.e., queries leading to abort answers) in between two successful queries. 
\par Once we prove that the simulator is straightline, we then show that the language ${\cal L}$ associated with the protocol has to be in BQP. This step is relatively more standard and follows similar ideas employed in prior works~\cite{GO94}. We use the straightline simulator to come up with a malicious prover $\widetilde{P}$ who is able to convince the verifier to accept an instance in ${\cal L}$ with probability close to 1; for now, assume that the protocol satisfies perfect completeness. At the same time, when given a NO instance, due to the soundness property, $\widetilde{P}$ can convince the verifier to accept only with negligible probability. Thus, using $\widetilde{P}$ and the honest verifier, we can construct a quantum polynomial time algorithm that decides ${\cal L}$. 

\paragraph{Forcing Simulator to be Straightline.} All that is left is to show that the simulator has to be straightline with respect to the contrived verifier $V'$ that we will design below. Since the simulator has complete control over the intermediate states of the verifier in-between different rounds of the protocol, it seems challenging to prevent the simulator from making multiple successful queries to the same channel $\Phi_i$, for some $i$, where $(\Phi_1,\ldots,\Phi_{k})$ are as defined above. We leverage the principles of quantum information to our advantage. Specifically, we use {\em subspace states}~\cite{AC13}, as intermediate private states of the verifier. Subspace states are known to be unclonable and have been instrumental in constructing public-key quantum money~\cite{AC13,Zhandry21}.   
\par In more detail, $V'$ is composed of the channels $(\Phi_1,\ldots,\Phi_n)$, where $\Phi_i$ works as follows: \begin{itemize}
    \item It takes as input a state $\rho_i$ and the $i^{th}$ round message from the prover.  
    \item It checks if $\rho_i$ is indeed the state $\ketbra{S_{i-1}}{S_{i-1}}$, where $S_{i-1}$ is a random $\frac{M}{2}$-dimensional subspace of $\mathbb{Z}_2^{M}$ and $\ket{S_{i-1}}=\sum_{x \in S_{i-1}} \sqrt{2^{-|S_{i-1}|}} \ket{x}$. The description of the subspace $S_{i-1}$ is hardwired in $\Phi_i$. 
    \item If the check fails, $\Phi_{i}$ aborts. 
    \item Otherwise, $V'$ computes the next message using the honest verifier and outputs this message along with another subspace state $\ket{S_i}$, where $S_i$ is a random $\frac{M}{2}$-dimensional subspace of $\mathbb{Z}_2^{M}$. The description of $S_i$ is hardwired in $\Phi_i$ (and $\Phi_{i+1}$). 
\end{itemize}
\noindent Let us look at the different possible queries a simulator can submit to $(\Phi_1,\ldots,\Phi_k)$. \\

\noindent {\em Out-of-order queries}: If the simulator tries to query $V'$ out-of-order, for example, submit a query to $\Phi_{i+1}$ before querying $\Phi_{i}$, it would not have the valid subspace state $\ket{S_{i}}$ and thus, $\Phi_{i+1}$ will most likely abort.\\

\noindent {\em Repeated queries}: If the simulator queries the same $\Phi_i$ twice then we claim that there can only be at most one query that passes the subspace state test. If the simulator makes two successful queries to $\Phi_i$ then we argue that it must have two copies of $\ket{S_{i-1}}$, thus violating the unclonability property of subspace states. We first argue this for $i=1$: this follows from the fact that the simulator has only one copy of $\ket{S_0}$. Assuming that at most only successful query can made to all $\Phi_j$, for $j <i$, we can then invoke the unclonability of $\ket{S_{i-1}}$ to argue that there can be at most one successful query to $\Phi_i$. 
\par One has to be careful when violating the unclonability property: the simulator can make many unsuccessful queries before making the two successful queries to $\Phi_i$. The reduction that violates the hardness of subspace states needs to be able to discern whether the adversary has submitted a successful query or not. Moreover, the reduction only receives as input a single copy of $\ket{S_{i-1}}$ and in particular, does not receive as input a description of $S_{i-1}$. Thus, the reduction does not have the capability to distinguish successful versus unsuccessful queries. We overcome this issue by giving the reduction access to a oracle that tests if a given state is $\ket{S_{i-1}}$. It was shown by~\cite{AC13} that as long as the subspace has super-logarithmic dimension, unclonability still holds even if the reduction gets access to the membership oracle.   
\\

\noindent {\em Missing queries}: There might exist some $i$ such that the simulator never submits a successful query to $\Phi_i$. By using similar arguments as in the previous cases, we can then argue that it cannot submit a successful query to the last round, i.e., $\Phi_k$. From this, we can come up with a distinguisher that can distinguish the real world and the ideal world. \\

\noindent From the above cases, we can argue that a successful simulator can only submit successful queries {\em in-order} and moreover, can only submit exactly one successful query for $\Phi_i$, for every $i \in [k]$. \\

\noindent However, there are a couple of issues we need to take into consideration while designing the straightline simulator to ensure that $\tilde{P}$ is successful. Firstly, we need to modify $V'$ so that it encrypts its internal state with a symmetric-key encryption scheme, whose secret key is also hardcoded in all $\Phi_i$. Otherwise, the straightline simulator would expect the internal state in the clear and $\tilde{P}$ would be unable to provide this information to the simulator. Secondly, we need that the internal state chosen by the simulator on query $\Phi_i$ to be exactly the same as the internal state after the query to $\Phi_{i-1}$. To solve this issue, $V'$ uses a digital signature scheme to sign its classical internal state. This forces the simulator to use the internal state chosen by the simulator on query $\Phi_i$ to be the internal state after the query to $\Phi_{i-1}$, otherwise we could violate the unforgeability of the digital signature scheme.

\subsection{Related Work}\label{sec:related-work}

\paragraph{Knowledge Tightness and Precise Zero-Knowledge.} Analogous questions have been explored in the classical literature. However, the main focus has been on designing zero-knowledge simulators whose runtime is closely related to the runtime of the verifiers. 

Goldreich, Micali and Widgerson~\cite{GMW86} (also,~\cite{Goldreich2001}) explore the notion of knowledge tightness which is the ratio of the running time of the simulator and the running time of the verifier. Some existing protocols~\cite{GMW,GMR,Blu86} already achieve constant knowledge tightness. Along the lines of knowledge tightness, Pass and Micali~\cite{MP06} formalized a notion called precise zero-knowledge where the simulator's runtime is closely related to the runtime of the verifier. Many followup works~\cite{PPSTV08,CPT12,DG12} study precise zero-knowledge in different settings. 
\par Of relevance is the work by Ning and Du~\cite{DG11} who consider a stronger notion of precise zero-knowledge, where the simulator's space complexity also needs to be closely related to the space complexity of the verifier. However, they study this notion only in the context of zero-knowledge against classical verifiers. 

\paragraph{Space-bounded verifiers.} Another related model considers proof systems where the verifier, even the honest one, has bounded space. Some works~\cite{AF93,SPY92} have proposed feasibility results in this model. We remind the reader that in our proof system, we don't limit the amount of classical space a malicious QPT verifier can have.

\paragraph{Post-Quantum Zero-Knowledge.} Watrous~\cite{Wat09} presented the first construction of zero-knowledge for NP against quantum verifiers. In the past few years, we have seen remarkable progress in understanding the feasibility of QZK. Bitansky and Shmueli~\cite{BS20} presented the first construction of quantum zero-knowledge in constant rounds. They prove that their construction satisfies QZK in the non-black box simulation setting. In fact, it was shown by Chia, Chung, Liu and Yamakawa~\cite{CCLY21} that constant round black box quantum zero-knowledge is impossible. This negative result was circumvented by a recent work by Lombardi, Ma and Spooner~\cite{LMS21} who achieved black-box constant round QZK in the coherent simulation setting. Relaxations of QZK systems have been considered in some recent works~\cite{AL20,CCY21}. Post-quantum zero-knowledge under composition have been studied~\cite{JKMR06,ABGKM20,ACP21}. Finally, zero-knowledge for quantum complexity classes such as QMA have also been studied~\cite{BJSW16,BG19}.

\subsection{Future Directions} 
In this section, we discuss future directions regarding fine-grained notions of quantum zero-knowledge.

\paragraph{Non-black-box quantum space-bounded simulation with superlogarithmic qubits.} Our result in \Cref{sec:impossibility} shows that fully black-box $(\omega(\log\lambda),\cdot)$-space-bounded quantum zero-knowledge is impossible, and as we discussed in \Cref{sec:def-space-bounded-zk}, (standard) black-box $(\cdot, p)$-space bounded quantum zero-knowledge is impossible for any polynomial $p$. We leave as an open problem to explore the feasibility of non-black black-box quantum $(\omega(\log\lambda),p)$-space-bounded simulation for some polynomial $p$.

\paragraph{New definitions of ``NISQ''-safe zero-knowledge protocols.} In this work, we initiated the study of classical zero-knowledge protocols, whose security is still guaranteed against {\em weaker} quantum adversaries.\footnote{As we discussed, since the adversaries are weaker, we can morally break the zero-knowledge property that have much more power than the verifier.} We focus here on the number of qubits allowed in the simulation and we leave as an open question the proposal of other meaningful definitions such as verifiers/simulators that use a fixed non-universal gateset or even a fixed architecture.

\paragraph{Space-bounded/fully black box simulation.} Another interesting direction is to explore the new notions of simulation that we introduce in this work in the context of other cryptographic protocols. Specifically, we could explore the feasibility of fully black-box simulation and space-bounded simulation in both classical protocols and also protocols that use quantum resources such as quantum oblivious transfer and quantum secure computation protocols~\cite{HSS11,ABGKM20,GriloL0V21,BartusekCKM21a}. 
\section{Preliminaries}
\noindent We denote the security parameter by $\secparam$. We assume basic familiarity of cryptographic concepts. 
\par We denote (classical) computational indistiguishability of two distributions $\distr_0$ and $\distr_1$ by $\distr_0 \approx_{c,\varepsilon} \distr_1$. In the case when $\varepsilon$ is negligible, we drop $\varepsilon$ from this notation. We denote the process of an algorithm $A$ being executed on input a sample from a distribution $\distr$ by the notation $A(\distr)$.

\subsection{Quantum Preliminaries}
Let $\cH$ be any finite Hilbert space, and let $L(\cH):=\{\cE:\cH \rightarrow \cH \}$ be the set of all linear operators from $\cH$ to itself (or endomorphism). Quantum states over $\cH$ are the positive semidefinite operators in $L(\cH)$ that have unit trace. %

A state over $\cH=\mathbb{C}^2$ is called a qubit. For any $n \in \mathbb{N}$, we refer to the quantum states over $\cH = (\mathbb{C}^2)^{\otimes n}$ as $n$-qubit quantum states. To perform a standard basis measurement on a qubit means projecting the qubit into $\{\ket{0},\ket{1}\}$. A quantum register is a collection of qubits. A classical register is a quantum register that is only able to store qubits in the computational basis.

\paragraph{Quantum Circuits.} A unitary quantum circuit is a sequence of unitary operations (unitary gates) acting on a fixed number of qubits. Measurements in the standard basis can be performed at the end of the unitary circuit. A (generalized) quantum circuit is a unitary quantum circuit with $2$ additional operations: $(1)$ a gate that adds an ancilla qubit to the system, and $(2)$ a gate that discards (trace-out) a qubit from the system. A non-uniform quantum polynomial-time algorithm (QPT) $C$ consists of a family $\{(C_n,\rho_n)\}_{n \in \mathbb{N}}$, where $C_n$ is an $n$-input qubit generalized quantum circuit of size $p(n)$ for some polynomial $p(\cdot)$, and $\rho_n$ is a density matrix assigned to a subset of input qubits of $C_n$. Unless explicitly specified, all the algorithms considered in this work are non-uniform algorithms.
\par Later, in~\Cref{sec:comp:model}, we present a more rigorous definition of quantum circuits that carefully describes the interaction between the classical and the quantum memory. 

\paragraph{Oracle Access.} A QPT algorithm $A$ has oracle access to a circuit $\Phi$, denoted by $A^{\Phi}$, if the computation proceeds as follows. On input a state $\rho$, defined on two registers ${\bf X}$ and ${\bf Y}$, $A^{\Phi}$ first applies a circuit on $\rho$, then applies $\Phi$ on the ${\bf Y}$ register, followed by applying another circuit on the result, followed by applying $\Phi$ again on the ${\bf Y}$ register and so on. 

\subsection{Computational Indistinguishability}

\noindent The following definition is due to~\cite{Wat09}.

\begin{definition}[Computational Indistinguishability of Quantum States] Let $I$ be an infinite subset $I \subset \{0,1\}^*$, let $p : \mathbb{N} \rightarrow \mathbb{N}$ be a polynomially bounded function, and let $\rho_{x}$ and $\sigma_x$ be $p(|x|)$-qubit states. We say that $\{\rho_{x}\}_{x \in I}$ and $\{\sigma_x\}_{x\in I}$ are \textbf{quantum computationally indistinguishable collections of quantum states} if for every QPT $\cE$ that outputs a single bit, any polynomially bounded  $q:\mathbb{N}\rightarrow \mathbb{N}$, and any auxiliary collection of $q(|x|)$-qubits states $\{\nu_x \}_{x \in I}$, and for all (but finitely many) $x \in I$, we have that
$$\left|\Pr\left[\cE(\rho_x\otimes \nu_x)=1\right]-\Pr\left[\cE(\sigma_x \otimes \nu_x)=1\right]\right| \leq \epsilon(|x|) $$
for some negligible function $\epsilon:\mathbb{N}\rightarrow [0,1]$. We use the following notation 
$$\rho_x \approx_{\quant,\epsilon} \sigma_x$$
and we ignore the $\epsilon$ when it is understood that it is a negligible function.
\end{definition}

\subsection{Interactive Protocols}
\paragraph{Languages and Relations.} A language $\lang$ is a subset of $\{0,1\}^*$. A relation $\rel$ is a subset of $\{0,1\}^* \times \{0,1\}^*$. We use the following notation:
\begin{itemize}

\item Suppose $\rel$ is a relation. We define $\rel$ to be {\em efficiently decidable} if there exists an algorithm $A$ and fixed polynomial $p$ such that $(x,w) \in \rel$ if and only if $A(x,w)=1$ and the running time of $A$ is upper bounded by $p(|x|,|w|)$. 

\item Suppose $\rel$ is an efficiently decidable relation. We say that $\rel$ is a NP relation if $\lang(\rel)$ is a NP language, where $\lang(\rel)$ is defined as follows: $x \in \lang(R)$ if and only if there exists $w$ such that $(x,w) \in \rel$ and $|w| \leq p(|x|)$ for some fixed polynomial $p$. 

\end{itemize}

\paragraph{Interactive Models.} We consider (classical) interactive protocols between two parties, a prover $P$ and a verifier $V$. We only consider interactive protocols corresponding to NP relations. An interactive protocol for an NP relation $\rel$ has the following format: a probabilistic polynomial-time prover $P$ takes as input an NP instance $x$ and a witness $w$ while a probabilistic polynomial time verifier $V$ takes as input $x$. Both the parties exchange some messages with each other. At the end of the protocol, the verifier $V$ outputs either Accept or Reject.

\paragraph{Notation.} We use the following notation in the rest of the paper.
\begin{itemize}
    \item $\langle P,V \rangle $ denotes the interactive protocol between $P$ and $V$. We denote the $\langle P(y_1),V(y_2) \rangle$ to be $(z_1,z_2)$, where $z_1$ is the prover's output and $z_2$ is the verifier's output. Sometimes we omit the prover's output and write this as $z \leftarrow \langle P(y_1),V(y_2) \rangle$ to indicate the output of the verifier to be $z$.
    \item $\view_V\left(\langle P(y_1),V(y_2) \rangle\right)$ denotes the view of the $V$ in the protocol $\Pi$, where $y_1$ is the input of $P$ and $y_2$ is the input of $V$. The view includes the output of $V$ and the transcript of the protocol.
    
\end{itemize}

\subsubsection{Proof and Argument Systems}
We start by recalling the definitions of the completeness and soundness properties of a classical interactive proof system. 

\begin{definition}[Proof System]
\label{def:pfsystem}
Let $\Pi$ be an interactive protocol between a classical PPT prover $P$ and a classical PPT verifier $V$. Let $\rel(\lang)$ be the NP relation associated with $\Pi$.
\par $\Pi$ is said to be a proof system if it satisfies the completeness and the soundness properties defined below.  
\begin{itemize} 
\item {\bf Completeness}: For every $(x,w) \in \rel(\lang)$, 
$$\prob[\accept \leftarrow \langle P(x,w),V(x) \rangle] \geq 1 - \negl(|x|),$$
for some negligible function $\negl$. 
\item {\bf Soundness}: For every  prover $P^*$ (possibly computationally unbounded), every $x \notin \rel(\lang)$, every state $\rho$,
    $$\prob\left[ \accept \leftarrow \langle P^*(x,\rho),V(x) \rangle  \right] \leq \negl(|x|),$$
   for some negligible function $\negl$. 
\end{itemize}
\end{definition}

\noindent We also consider the notion of argument systems where the prover is restricted to be a QPT algorithm. 

\begin{definition}[Argument System]
Let $\Pi$ be an interactive protocol between a classical PPT prover $P$ and a classical PPT verifier $V$. Let $\rel(\lang)$ be the NP relation associated with $\Pi$.
\par $\Pi$ is said to be an argument system if it satisfies completeness (as defined in \Cref{def:pfsystem}) and computational soundness (defined below). 
\begin{itemize}
    \item {\bf Computational Soundness}: For every QPT prover $P^*$, every $x \notin \rel(\lang)$, for every $\poly(|x|)$-qubit state $\rho$,
    $$\prob\left[ \accept \leftarrow \langle P^*(x,\rho),V(x) \rangle  \right] \leq \negl(|x|),$$
   for some negligible function $\negl$. 
\end{itemize}
\end{definition}

\subsubsection{Quantum  Witness-Indistinguishable Proofs for NP} For our construction, we use a proof system that satisfies a property called quantum witness indistinguishability. We recall this notion below. 
\begin{definition}[Quantum Witness-Indistinguishability] 
\label{def:qwi}
An interactive protocol between a (classical) PPT prover $P$  and a (classical) PPT verifier $V$ for a language $L \in \np$ is said to be a \textbf{quantum witness-indistinguishable proof system} if in addition to completeness, unconditional soundness, the following holds: 
\begin{itemize}
    \item {\bf Quantum Witness-Indistinguishability}: For every QPT verifier $V^*$, for every $x \in \lang$, $w_1,w_2$ such that $(x,w_1) \in \rel(\lang)$ and $(x,w_2) \in \rel(\lang)$,  with $\poly(|x|)$-qubit advice $\rho$, the following holds: 
    $$\left\{ \view_{\vrfr^*}\left(\langle P(x,w_1),V^*(x,\rho) \right) \right\} \approx_{\quant} \left\{ \view_{\vrfr^*}\left(\langle P(x,w_2),V^*(x,\rho) \right) \right\}$$
\end{itemize}
\end{definition}

\paragraph{Instantiation.} By suitably instantiating the constant round WI argument system of Blum~\cite{Blu86} with statistically binding commitments (which in turn can be based on post-quantum one-way functions~\cite{Naor91}), we achieve a 4 round quantum WI proof system for NP. Moreover, this proof system is a public-coin proof system; that is, the verifier's messages are sampled uniformly at random. 

\subsection{Statistically Binding and Quantum-Concealing Commitments}
\label{sec:prelims:commit}

\noindent A bit commitment protocol is a two-party protocol defined between two parties, a committer ($\comm$) and a receiver ($\receiver$). 
\par The protocol consists of two phases:
\begin{itemize}
    \item Commit phase: in this phase, $\comm$ commits to a bit $b$ and $\receiver$ does not receive any input.\\ {\em In this work, we only consider commitments with two-message commit phases}. 
    \item Opening phase: in this phase, $\comm$ reveals $b$ and any relevant randomness used in the commit phase. $\receiver$ then outputs accept or reject.\\ {\em In this work, we consider only canonical opening phases. That is, they only consist of a single message. Moreover, the opening consists of the revealed bit and all the random bits used by $\comm$.} 
\end{itemize}

\noindent As a result of considering canonical opening phases, we have the following property, referred to as {\em message recovery} property: given just the randomness used by the committer, the receiver can recover the bit used in the commit phase. 

\paragraph{Properties.} We consider a two-message commitment scheme that satisfies the following two properties. 
\begin{definition}[Statistically Binding]
A two-message commitment scheme between a committer ($\comm$) and a receiver ($\receiver$), both running in probabilistic polynomial time, is said to satisfy statistical binding property if the following holds for any adversary $\adversary$: 
$$\prob \left[ \substack{(\bfc,r_1,x_1,r_2,x_2)  \leftarrow \adversary\\ \ \\ \text{ and }\\ \  \\  \comm(1^{\secparam},\bfr,x_1;r_1) = \comm(1^{\secparam},\bfr,x_2;r_2) = \bfc\\  \ \\ \text{ and }\\ \ \\ x_1 \neq x_2}: \bfr \leftarrow \receiver(1^{\secparam}) \right] \leq \negl(\secparam),$$
for some negligible function $\negl$. 
\end{definition}

\begin{definition}[Quantum-Concealing]
A commitment scheme $\comm$ is said to be quantum concealing if the following holds. Suppose $\adversary$ be a non-uniform QPT algorithm and let $\bfr$ be the message generated by $\adversary(1^{\secparam})$. We require that $\adversary$ cannot distinguish the two distributions,  $\{\comm(1^{\secparam},\bfr,x_1)\}$ and $\{\comm(1^{\secparam},\bfr,x_2)\}$, for any two inputs $x_1 \in \{0,1\},x_2 \in \{0,1\}$. 
\end{definition}

\paragraph{Instantiation.} We can instantiate statistically binding and quantum-concealing commitments from post-quantum one-way functions~\cite{Naor91}. 

\paragraph{Generalizations.} We can extend the above definition to commitments of long messages and not just bits. We consider a specific instantiation of commitments of long messages. To commit to a long message, the committer commits to each bit separately. Note that even this instantiation satisfies the message recovery property.

\subsection{Query Lower Bounds for Cloning Subset States}
\label{sec:querylb}
\noindent In the impossibility result (\Cref{sec:impossibility}), we use the following theorem stated and proved by Aaronson and Christiano~\cite{AC13} in the context of public-key quantum money schemes.  

\paragraph{Subspace states.} For a subspace $A \subseteq \mathbb{F}_2^n$, for some $n \in \mathbb{N}$, we denote $\ket{A}=\sum_{x \in A} \frac{1}{\sqrt{2^{|A|}}} \ket{x}$. 
We also denote $U_A$ to be the following unitary. 
It maps an $n+1$-qubit state $\ket{A}\ket{y}$ into $\ket{A}\ket{y \oplus 1}$, and for every $\ket{\psi}$ orthogonal to $\ket{A}$, it leaves $\ket{\psi}\ket{y}$ unchanged. Notice that $U_A$ can be efficiently implemented with the description of $A$. 

\paragraph{Space complexity of testing subspace states.} In \Cref{sec:impossibility}, we will need to implement the following projective measurement $\Pi = \{\underbrace{\kb{A}}_{\text{outcome}=0}, \underbrace{I-\kb{A}}_{\text{outcome}=1}\}$, where $A \subseteq \mathbb{F}_2^n$ is a subspace of dimension $k$, with low quantum memory complexity. 
\par We will demonstrate an efficient implementation of $\Pi$ that takes space $O(n)$ qubits. In order to implement $\Pi$, we first describe how to implement an unitary $C_A$ that maps $\ket{0^k}\ket{0^n} \mapsto \ket{0^k}\ket{A}$, given some basis $v_1,\ldots,v_k$ of the subspace $A$. We denote the first $k$ qubits to be the register ${\bf X}$ and the next $n$ qubits to be the register ${\bf Y}$. Notice that we can assume without loss of generality, we can assume that $v_i$ and $v_j$ are orthogonal for $i \ne j$. 

$C_A$ can be implemented as follows: 
\begin{enumerate}
    \item Apply $H^{\otimes k}$ to the first register ${\bf X}$,
    \item For each $i = 1,\ldots,k$:
    \begin{enumerate}[label*=\arabic*.]
        \item On the $i$-th qubit of ${\bf X}$ and $i$-th qubit of ${\bf Y}$, perform the following operation:
        $ \ket{b}\ket{x} \mapsto \ket{b}\ket{b\cdot v_i \oplus x}$
    \end{enumerate}
    \item For each $i = 1,\ldots,k$:
    \begin{enumerate}[label*=\arabic*.]
        \item With the $i$-th qubit of ${\bf X}$ and $i$-th qubit of ${\bf Y}$, perform the following operation:
        $ \ket{b}\ket{x} \mapsto \ket{b \oplus \langle x, v_i \rangle}\ket{x},$
        where $\langle y,z \rangle$ is the inner product of $y,z\in \mathbb{F}_2^n$.
    \end{enumerate}    
\end{enumerate}
After step $2$, we have the state of the form $\sum_{{\bf b} \in \{0,1\}^k} 2^{-k/2} \ket{{\bf b}}\ket{\oplus_{1 \leq i \leq k} {\bf b}_i v_i}$, and that step $3$ uncomputes the first register ${\bf X}$ (and here we use the crucial fact that the basis is orthogonal).

In order to finally perform the projective measurement $\{\underbrace{\kb{A}}_{\text{outcome}=0}, \underbrace{I-\kb{A}}_{\text{outcome}=1}\}$ on some state $\rho$, we can then simply prepend the state $\ket{0^k}$ to it, perform the unitary $C_A^{\dagger}$, trace out ${\bf X}$ and then measure the last register in the computational basis. If we get the value $0^n$ then the outcome of $\Pi$ is set to be $0$, otherwise it is set to be $1$. 
\par To see, why this is the correct implementation, notice that $\Tr_{{\bf X}}\left( C_{A} \ket{0^k} \ket{0^n} \bra{0^k} \bra{0^n} C_A^{\dagger} \right) = \Tr_{{\bf X}}(\ketbra{0^k}{0^k} \otimes \ketbra{A}{A}) = \ketbra{A}{A}$. 
\par We notice that to implement this projective measurement, we need only $k$ extra qubits. Thus, the total memory complexity needed to implement $\Pi$ is $n+k$ qubits. 

\paragraph{Testing subspace states.} Throughout this work, we will consider the test ${\sf Test}^{U_A}(\rho)$ to check if $\rho$ is a subspace state $\ket{A}$ when we have oracle access to $U_A$:
\begin{enumerate}
    \item Query the oracle $U_A$ on input $\rho \otimes \kb{0}$.
    \item Measure the last qubit, with output $b$. Let $\sigma_b$ be the reduced state on the first register after the measurement.
    \item Return $\kb{b} \otimes \sigma_b$.
\end{enumerate}

\noindent Whenever the first bit is $\kb{1}$, we say that the test passes, otherwise we say that the test fails.

We have the following lemma regarding the resulting state of the verification test.
\begin{lemma}
\label{lem:uncstates:test}
The state $\sigma_1$ corresponding to the quantum state returned when ${\sf Test}^{U_A}(\rho)$ passes is equal to $\kb{A}$.
\end{lemma}

\noindent The proof of this lemma follows directly from the definition of $U_A$.

\paragraph{Unclonability of subspace states.} We state now the result from \cite{AC13} regarding the unclonability of subspace states. For that, let ${\cal S}$ be the set of all subspaces of $\mathbb{F}^n_2$ of dimension $\frac{n}{2}$.

\begin{lemma}[Theorem 25,~\cite{AC13}]
\label{lem:money}
Let $A \subseteq \mathbb{F}^n_2$
be a subspace sampled uniformly at random from $\mathcal{S}$. Then given $\ket{A}$, as well as oracle access to $U_A$, any quantum algorithm $C$ needs
$\Omega(\sqrt{\eps}2^{n/4})$ queries to prepare a $2n$-qubit state $\rho$ that projects into $\kb{A}^{\otimes 2}$ with probability (over choice of $A$, $C$ and the projection onto $\kb{A}^{\otimes 2}$) at least $\eps$, for all
$\eps = \omega(2^{-n/2})$.  
\par More formally, $C$ needs $\Omega(\sqrt{\eps}2^{n/4})$ queries to prepare a $2n$-qubit state $\rho$ such that,\\ $\sum_{A \in {\cal S}} \frac{1}{|{\cal S}|} \Tr(\ketbra{A}{A}^{\otimes 2} C^{U_A}(\ketbra{A}{A})) \geq \epsilon$.
\end{lemma}

\section{Space-bounded Quantum Zero-Knowledge: Definitions}
We present the definition of quantum zero knowledge with space bounded simulation. First, we will revisit the definition of zero-knowledge in~\Cref{sec:fully-black-box}. Towards defining space-bounded simulation, we present the computational model in~\Cref{sec:comp:model}. This model will explicitly capture the interaction between the classical memory and the quantum memory. We then define space-bounded quantum zero-knowledge in \Cref{sec:def-space-bounded-zk}.   

\newcommand{\purify}{\mathscr{P}}
\newcommand{\channel}{\mathscr{C}}
\subsection{Defining Quantum Zero-Knowledge}
\label{sec:fully-black-box}
\noindent We discuss the different types of quantum zero-knowledge below. We first start with the notion of quantum zero-knowledge, referred to as {\em black-box} quantum zero-knowledge, considered in the literature.

\paragraph{Black-Box QZK.} Consider the following two definitions. 

\begin{definition}[Canonical purification]
\label{def:can:purification}
Let $V$ be a quantum channel that is implemented with quantum gates, intermediate measurements and trace-out operations. We denote the canonical purification of $V$ as $\purify(V)$, which consists of the unitary circuit that performs all the measurements coherently and replaces all the trace-out operations with identity gates.
\end{definition}

\noindent Any circuit can be expressed as its purification followed by measurements and trace out operations.

\begin{definition}
\label{def:notation:channel}
Let $V$ be a party in a $k$-round interactive protocol. We define  $\channel(V) = (\Phi_1,...,\Phi_k)$, where $\Phi_i : \mathcal{M}_i \times \mathcal{V}_i \to \mathcal{M}_{i+1} \times \mathcal{V}_{i+1}$ to be the quantum channel corresponding to the computation performed by $V$ on the $i^{th}$ round of the protocol, where $\mathcal{V}_i$ is the Hilbert space corresponding to the memory of $V$ on round $i$ and $\mathcal{M}_i$ is the Hilbert space corresponding to message sent in round $i$. 
\par We denote $\purify(V)=(\purify(\Phi_1),\ldots,\purify(\Phi_k))$, where $\channel(V)=(\Phi_1,\ldots,\Phi_k)$ and $\purify$ is as defined in~\Cref{def:can:purification}. Similarly, we define $\purify(V)^{\dagger}=\left(\purify(\Phi_1)^{\dagger},\ldots,\purify(\Phi_k)^{\dagger}\right)$. 
\end{definition}

\noindent Most of the recent works~\cite{Wat09,ACP21,CCY21,CCLY21,LMS21,CCLY22} consider the following definition of quantum zero-knowledge.

\begin{definition}[Black-Box Quantum Zero-Knowledge]
\label{def:standard-bb}
A proof (resp., argument) system $(P,V)$ for an NP relation $\rel$ is black-box quantum zero-knowledge if there exists a QPT simulator $\Sim$ such that for every QPT verifier $V^*$, for every $(x,w) \in \rel$, for every $\poly(|x|)$-qubit state $\rho$, 
$$ \view_{V^*}\left(\langle P(x,w),V^*(x,\rho) \rangle\right) \approx_{\quant,\eps(|x|)} {\Sim^{{\purify(V^*),\purify(V^*)^{\dagger}}}}(x,\rho),$$
where $\Sim$ having oracle access to $\purify(V^*)$ and $\purify(V)^{\dagger}$ is denoted by $\Sim^{\purify(V^*),\purify(V^*)^{\dagger}}$, $\purify(\cdot)$ is as defined in~\Cref{def:can:purification,def:notation:channel} and $\eps(|x|)$ is a negligible function in $|x|$.
\end{definition}

\noindent {\em Comparison with Classical Black-Box ZK}: In the definition of black-box classical zero-knowledge, the simulator has oracle access to the verifier circuit $V$. However, in~\Cref{def:standard-bb}, the simulator does not have direct access to the verifier's circuit. Instead it has oracle access to the purification of the verifier circuit. Not only that, it also has oracle access to the inverse of the purification as well. Thus, we believe that the above definition is not the true quantum analogue of classical black-box zero-knowledge.

\paragraph{Fully Black-Box QZK.} We consider a new definition that resembles the black-box definition in the classical setting. We call this definition fully black-box QZK. Just like the classical setting, the simulator in the fully black-box QZK only gets oracle access to the verifier circuit, and not the purified version as considered in~\Cref{def:standard-bb}. 

We now present our definition of fully black-box zero-knowledge.

\begin{definition}[Fully Black-Box Quantum Zero-Knowledge] \label{def:fully-bb}
A proof (resp., argument) system $(P,V)$ for an NP relation $\rel$ is fully black-box quantum zero-knowledge if there exists a QPT polynomial-time simulator $\Sim$ such that for every QPT malicious verifier $V^*$, for every $(x,w) \in \rel$, for every state $\rho$,
$$ \view_{V^*}\left(\langle P(x,w),V^*(x,\rho) \rangle\right) \approx_{\quant,\eps(|x|)} \Sim^{\channel(V^*)}(x,\rho),$$
where $\Sim$ having oracle access to $\channel(V^*)$ is denoted by $\Sim^{\channel(V^*)}$, $\channel(\cdot)$ is as defined in~\Cref{def:notation:channel} and $\eps(|x|)$ is negligible in $|x|$. 
\end{definition}

\paragraph{Non Black-Box QZK.} We can also define the most general definition of QZK where the simulator can depend arbitrary on the code of the verifier. We call this non black box QZK. This definition was considered in a few recent works~\cite{BS20,AL20,BKS21}.

\begin{definition}[Non Black-Box Quantum Zero-Knowledge] \label{def:non-bb}
A proof (resp., argument) system $(P,V)$ for an NP relation $\rel$ is non black-box quantum zero-knowledge if for every QPT verifier $V^*$, there exists a QPT polynomial-time simulator $\Sim$, for every $(x,w) \in \rel$, for every $\poly(|x|)$-qubit state $\rho$,
$$ \view_{V^*}\left(\langle P(x,w),V^*(x,\rho) \rangle\right) \approx_{\quant,\eps(|x|)} \Sim(x,\rho),$$
where $\eps(|x|)$ is negligible in $|x|$. 
\end{definition}

\subsection{Computational model}
\label{sec:comp:model}
\noindent Towards defining space-bounded simulation, we define the model of quantum computation albeit using a different terminology. The new terminology is necessary to carefully model the interaction between the classical memory and the quantum memory since we are explicitly measuring the amount of quantum memory required for any computation. 
\par First, we start with a weaker computational model. The following definition considers quantum circuits that take as input a unitary and a state and output the computation of the unitary on the state. 

\begin{definition}[$M$-qubit programmable quantum computer]
We say that $Q$ is an $M$-qubit programmable quantum computer, if $Q$ receives as an input an $M$-qubit quantum state $\ket{\psi}$, the (classical) description of an $M$-qubit unitary $U$ and it outputs the state $U\ket{\psi}$. We depict such a quantum device in \Cref{fig:qdevice}.
\end{definition}

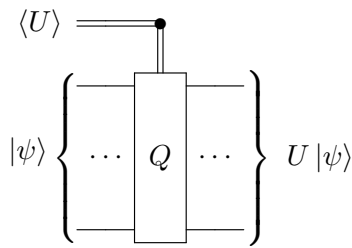
\begin{figure}[!htb]
$$
\Qcircuit @C=1em @R=1.6em {
 \lstick{\langle U \rangle} & \cw & \control \cw  & &  \\
\lstick{} & \qw & \multigate{2}{Q}  \cwx & \qw & \qw &  \\
\lstick{} & \cdots & \nghost{Q} &  \cdots & & \rstick{U\ket{\psi}} \\
\lstick{} & \qw & \ghost{Q} & \qw  & \qw  \inputgroupv{2}{4}{.8em}{2.3em}{\ket{\psi}} \gategroup{2}{5}{4}{5}{1em}{\}}\\
}
$$
\caption{$M$-qubit programmable quantum computer}
\label{fig:qdevice}
\end{figure}

\noindent We now generalize the above model in the following way. After the computation of the unitary, we allow for some classical preprocessing. This is enabled by measuring a subset of the qubits of the state obtained after the unitary computation and the resulting measurement outcomes can be used to choose the function to be applied on the classical memory. 
\par This entire computation is considered to be a block of computation. Later, we will consider the most general model where the computation is defined to a sequence of blocks. 

\begin{definition}[$M$-qubit of unitary quantum computation with classical post-processing]
\label{def:block}
An $M$-qubit of unitary quantum computation with classical post-processing is defined by a tuple $(c,\langle U \rangle, \ell, \ket{\psi}, f)$, where $c \in \{0,1\}^{L}$ is a classical register, $\langle U \rangle$ is the description of an $M$-qubit unitary, $\ell \in [M+1]$ is the number of qubits that will be measured, $\ket{\psi}$ is the input state of the quantum computation and $f$ is a classical function used for post-processing.

The computation proceeds as follows:
\begin{enumerate}
    \item Using an $M$-qubit programmable quantum computer $Q$, we run the unitary $U$ on the state $\ket{\psi}$,
    \item We measure the first $\ell$ qubits of the output, having outcome $o$ and the state on the last $M-\ell$ qubits is $\ket{\phi_o}$,
    \item We apply a classical function $f$ on $(c,\langle U \rangle, \ell, o)$, resulting in $(c', \langle U'\rangle, \ell')$,
    \item The first $\ell$ qubits are reset to $\ket{0}$,
    \item The output of the computation is $(c', \langle U'\rangle, \ell', \ket{\psi'})$, where $\ket{\psi'} = \ket{0^\ell}\ket{\phi_o}$.
\end{enumerate}
\noindent We depict such a computation in \Cref{fig:block-computation}.
\end{definition}

\begin{figure}[!htb]
$$
\Qcircuit @C=1em @R=1.6em {
 \lstick{c} & \cw    &  \cw  &  \cw & \cw & \cw & \cmultigate{2}{f} & \cw & \cw &\cw & \cw & \rstick{c'} & & &  \\
 \lstick{\langle U \rangle} & \cw & \control \cw  & \cw & \cw & \cw & \cghost{f}   & \cw &\cw & \cw &  \cw & \rstick{\langle U' \rangle}  \\
 \lstick{\ell} & \cw & \cwx \cw  &  \cw  &  \cw & \control \cw & \cghost{f}   & \cw &\cw & \cw &  \cw & \rstick{\ell'}  \\
\lstick{} & \qw & \multigate{3}{Q}  \cwx & \qw & \qw & \meter \cwx   & \control \cw \cwx   & & & \ket{0}& & \qw   \\
\lstick{} & \qw & \ghost{Q} & \qw  & \qw & \meter & \control \cw \cwx & & & \ket{0} & & \qw & & \ket{\psi'} \\
\lstick{} & \cdots & \nghost{Q} &  \cdots & &  &  &  & & & \cdots \\
\lstick{} & \qw & \ghost{Q} & \qw  & \qw  & \qw  & \qw  & \qw  & \qw  & \qw  & \qw & \qw 
\inputgroupv{4}{7}{.8em}{2.3em}{\ket{\psi}} \gategroup{4}{12}{7}{12}{1em}{\}}
& & & & & B_i
\\
}
$$
\caption{$M$-qubit of unitary quantum computation with classical post-processing}
\label{fig:block-computation}
\end{figure}
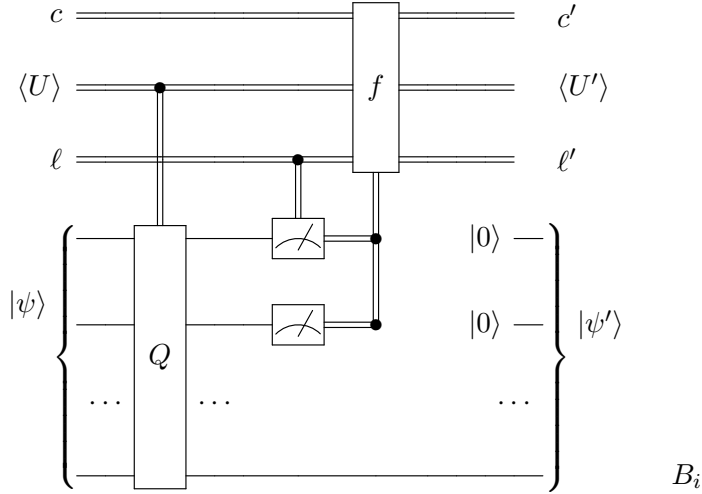

\paragraph{Our Model.} We are now ready to define the model that we consider in this work. As mentioned earlier, the model below comprises of a sequence of blocks of computation, where each block is an $M$-qubit of unitary computation with classical post-processing. 

\begin{definition}[$M$-qubit quantum channel]
\label{def:mqubitchannels}
An $M$-qubit of quantum computation with classical post-processing is defined by a tuple $(c_0,\langle U_0 \rangle, \ell_0, \ket{\psi_0}, f)$, where $c \in \{0,1\}^{L}$ is a classical register, $\langle U_0 \rangle$ is the description of an $M$-qubit unitary, $\ell_0 \in [M+1]$ is the number of qubits that will be measured, $\ket{\psi_0}$ is the input state of the quantum computation and for $f$ is a classical function used for post-processing.

The computation proceeds as follows:
We run an M-qubit of unitary quantum computation with classical post-processing (\Cref{def:block}) defined by $(c_i,\langle U_i \rangle, \ell_i, \ket{\psi_i}, f)$ and we let the output of the computation be $(c_{i+1}, \langle U_{i+1} \rangle, \ell_{i+1}, \ket{\psi_{i+1}})$. If $\langle U_{i+1} \rangle = \perp$, this process stops and output $c', \ell'$ and $\ket{\psi'}$.

\par We depict such a computation in~\Cref{fig:full-model}. 
\end{definition}

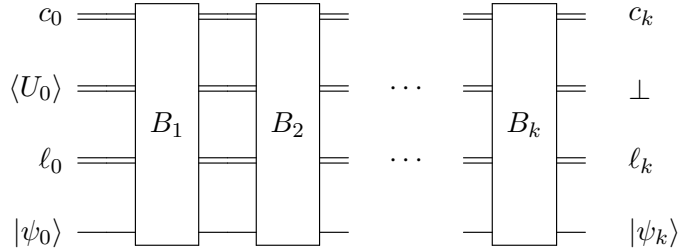
\begin{figure}[!htb]
$$
\Qcircuit @C=1em @R=1.6em {
 \lstick{c_0} & \cw    &  \cghost{B_1}  & \cw &  \cghost{B_2} & \cw & 
 &   & & 
 & \cghost{B_k} & \cw & \rstick{c_{k}} \\
 \lstick{\langle U_{0} \rangle} & \cw & \cghost{B_2}  & \cw &  \cghost{B_2} & \cw &  &  \cdots & & 
 & \cghost{B_k} & \cw & \rstick{\perp} \\
 \lstick{\ell_0} & \cw & \cghost{B_2}  & \cw &  \cghost{B_2} & \cw &  &  \cdots & & 
 & \cghost{B_k} & \cw & \rstick{\ell_{k}} \\
\lstick{\ket{\psi_0}} & \qw & \multigate{-3}{B_1} & \qw &  \multigate{-3}{B_2} & \qw &  &   & & 
 & \multigate{-3}{B_k} & \qw & \rstick{\ket{\psi_k}}
\\
}
$$
\caption{$M$-qubit quantum computation with classical post-processing. Each block $B_i$ corresponds to the circuit in \Cref{fig:block-computation} with the corresponding choices of input and classical post-processing $f_i$.}
\label{fig:full-model}
\end{figure}

\noindent We consider $M$-qubit channels that are efficiently implementable. 

\begin{definition}[$M$-qubit Quantum Polynomial Time Algorithm]
\label{def:qpt:space}
A (non-uniform) quantum polynomial algorithm $C$ is said to be an $M$-qubit (resp. non-uniform) strict quantum polynomial-time algorithm if it is an $M$-qubit quantum channel and there is a polynomial $p$ such that each of the blocks runs in time at most $p(|x|)$ and after $p(|x|)$ rounds, the procedure stops (i.e., there exists some $k \leq p(|x|)$ such that $\langle U_k \rangle = \perp$).

A (non-uniform) quantum polynomial algorithm $C$ is said to be an $M$-qubit (resp. non-uniform) expected quantum polynomial-time algorithm if it is an $M$-qubit quantum channel and there is a polynomial $p$ such that each of the blocks runs in time at most $p(|x|)$ and the expected number of rounds of the procedure is $p(|x|)$ (i.e., $\mathbb{E}[ k \; |   \langle U_k \rangle = \perp] \leq p(|x|)$, where the expectation is over the measurement outcomes of the procedure).

\end{definition}

\paragraph{Oracle Access.} In order to discuss space-bounded black-box zero-knowledge, we need to define how the to give oracle access to an $M$-qubit quantum channel and how the memory of the overall channel is counted.

An $M_A$-qubit quantum channel $A$ has oracle access to an $M_B$-qubit quantum channel $B$, for $M_A\geq M_B$, then $A$ prepares a  classical register ${\bf B_C}$ and $M_B$-qubit quantum register ${\bf B_Q}$, and then runs the channel $B$ on input registers ${\bf B_C}$ and ${B_Q}$. $A$ then has access to the output of $B$, consisting again of a classical register and an $M_B$-qubit quantum register.

\subsection{Space-Bounded Quantum Zero-Knowledge}
\label{sec:def-space-bounded-zk}

\noindent Equipped with the computational model described in~\Cref{sec:comp:model}, we state the definition of space-bounded quantum zero-knowledge. Roughly speaking, the definition states that the number of qubits consumed by the simulator should be a fixed polynomial in the number of qubits consumed by the verifier. Ideally, we would like the polynomial to be of low degree; however, we place no such restriction on the definition below. 

\begin{definition}[$(s, p)$-space-bounded QZK] Let ${\cal X} \in \{\text{fully black box},\text{black box},\text{non black box}\}$. An ${\cal X}$ QZK proof (resp., argument) system $(P,V)$ for a language $\lang$ is said to be $(s,\functionspace)$-space-bounded ${\cal X}$ QZK proof (resp., argument) system if the following property is satisfied:
\begin{itemize}

    \item Suppose $V^*$ is an $\ell_{V^*}$-qubit QPT (\Cref{def:qpt:space}), where $\ell_{V^*}=s(|x|)$. Then, ${\sf Sim}$ is an $\ell_{{\sf Sim}}$-qubit QPT, where $\ell_{{\sf Sim}} \leq f(s(|x|))$.  
\end{itemize}
\end{definition}

\noindent In the case of the black box QZK (\Cref{def:standard-bb}), the amount of quantum memory utilized by $\Sim$ is compared against the amount of quantum memory utilized by $V^*$ before its purification. Moreover, by definition, the quantum memory of $\Sim$ contains the memory of the purified verifier $V^*$. For any polynomial $\functionspace(\cdot)$, for any $s$, it is possible to come up with an $s$-qubit channel\footnote{For example, consider an $s$-qubit channel contains $\functionspace(s)$ bits of classical memory. After canonical purification, the overall amount of quantum memory would be more than $\functionspace(s)$ qubits. Note that we don't rule out other types of purification which critically use the structure of the $s$-qubit channel to not increase the quantum memory significantly. However, this is not in the spirit of black box QZK and we consider such variants to be closer to non black box QZK.} (\Cref{def:mqubitchannels}) such that purifying this channel would require strictly more than $\functionspace(s)$ qubits. This means that any black box simulator against a verifier represented as this $s$-qubit channel would utilize more than $\functionspace(s)$ qubits. Hence, $(s,\functionspace)$-space bounded black box QZK, for any polynomial $\functionspace$, is impossible to achieve. On the other hand, for superpolynomial $f(\cdot)$, existing works demonstrate that $(s,\functionspace)$-space bounded QZK exists .
\par From the above discussion, it follows that the only two meaningful notions of space bounded QZK to consider are fully black box QZK and non black box QZK.

\section{Zero-Knowledge against Logarithmic Quantum Space Verifiers}
\label{sec:positive}

\noindent We present the construction of an argument system $\Pi$ for the NP relation $R(\lang)$ which is space-bounded fully black-box zero-knowledge against  $O(\log(\secparam))$-qubit quantum channels (\Cref{def:mqubitchannels}). 
\par We use the following tools in our construction. 
\begin{itemize}
    \item Statistically-binding and quantum-concealing commitment protocol (see Section~\ref{sec:prelims:commit}), denoted by $(\comm,\receiver,\open)$.
    \item Four round quantum witness-indistinguishable proof system $\protwi$ (Definition~\ref{def:qwi}). The relation associated with $\protwi$, denoted by $\relwi$, is defined as follows: 
    $$\relwi = \Bigg\{ \left(\left( x,\bfc^*,\bfc^{**},\left\{\alpha_{i,b} \right\}_{i \in [\secparam],b\in \{0,1\}}\right);\ \left( w,\bfr^*,r,i^* \right) \right)\ :\ (x,w) \in \rel(\lang)\ \ \text{or}$$
    $$\open(\bfc^*)=\left( \bfr^*,r \right) \text{ and }\open(\bfc^{**})=\left( (i^*,\alpha_{i^*,0} \oplus \alpha_{i^*,1}),\bfr^* \right) \Bigg\}$$
\end{itemize}

\noindent We present the construction of $\Pi$ in~\Cref{fig:pcoinczk}.

\begin{figure}[!htb]
   \begin{center} 
   \begin{tabular}{|p{16cm}|}
    \hline \\
    {\bf Input of $P$}: Instance $x \in \lang$ along with witness $w$.  \\
    {\bf Input of $V$}: Instance $x \in \lang$. \\
    \ \\
    \noindent {\bf Trapdoor Phase}: 
    \begin{enumerate}
        \item $P \leftrightarrow V$: $P$ and $V$ engage in an execution of the commitment protocol, where $P$ plays the role of the committer and $V$ plays the role of the receiver. $P$ commits to $\bfr^*$, where $\bfr^* \xleftarrow{\$} \{0,1\}^{\secparam}$. Denote the commitment to be $\bfc^*$. Let the randomness used by $P$ in the commitment protocol be $r$. 
       
        \item $V \leftrightarrow P$: $P$ and $V$ run $2\secparam$ parallel executions of the commitment protocol, with $V$ playing the role of the committer and $P$ playing the role of the receiver. We index the executions using the notation $(i,b)$, where $i \in [\secparam]$ and $b\in \{0,1\}$. In the $(i,b)^{th}$ execution, $V$ commits to $\alpha_{i,b}$, where $\alpha_{i,b} \xleftarrow{\$}  \{0,1\}^{\secparam}$. Denote the commitments by $\left( \{\bfc_{i,b}\}_{i \in [\secparam],b\in \{0,1\}} \right)$.

        \item $P \rightarrow V$: $P$ sends $\bfb_1,\ldots,\bfb_{\secparam}$, where $\bfb_i \xleftarrow{\$} \{0,1\}^{\secparam}$. 
        
        \item $V \rightarrow P$: $V$ sends  the openings of all the commitments $(\{\bfc_{i,\bfb_i}\}_{i \in [\secparam]})$. Denote the opening of the $(i,\bfb_i)^{th}$ commitment to be $o_{i,\bfb_i}=\open(\bfc_{i,\bfb_i})$. If any of the openings are invalid, $P$ rejects. Otherwise, $P$ parses $o_{i,\bfb_i}$ as $(\alpha_{i,\bfb_i},r_{i,\bfb_i})$. 
        
        \item $P \leftrightarrow V$: $P$ and $V$ engage in another execution of the commitment protocol, where $P$ plays the role of the committer and $V$ plays the role of the receiver. $P$ commits to $0$, using the randomness $\bfr^*$. Denote the commitment to be $\bfc^{**}$. 
        
        \item $V \rightarrow P$: $V$ sends the rest of the openings of $\left( \{\bfc_{i,b}\}_{i \in [\secparam],b\in \{0,1\}} \right)$. That is, it sends $\{o_{i,1-\bfb_i}\}_{i \in [\secparam]}$, where $o_{i,1-\bfb_i}=\open(\bfc_{i,1-\bfb_i})$. If any of the openings are invalid, $P$ rejects. Otherwise, $P$ parses $o_{i,1-\bfb_i}$ as $(\alpha_{i,1-\bfb_i},r_{i,1-\bfb_i})$.   
 
        \end{enumerate}

    \ \\
    \noindent {\bf WI Phase}: $P$ and $V$ engage in $\protwi$ with the common input being the following: $$\bfx_{\sf wi}= \left( x,\bfc^*,\bfc^{**},\left\{\alpha_{i,b} \right\}_{i \in [\secparam],b\in \{0,1\}}\right)$$ Additionally, $P$ uses the witness  $(w,\bot,\bot,\bot)$. \\
    \ \\
    \hline
   \end{tabular}
    \caption{Description of the protocol $\Pi$.}
    \label{fig:pcoinczk}
    \end{center}
\end{figure}

\paragraph{Completeness.}\ The completeness of the protocol $\Pi$ in~\Cref{fig:pcoinczk} follows from the completeness of $\protwi$.

\subsection{Proof of Soundness}
\begin{proposition}
The protocol $\Pi$ in~\Cref{fig:pcoinczk} satisfies soundness.
\end{proposition}
\begin{proof}
Let $P^*$ be a (non-uniform) malicious prover. Let $x \in \{0,1\}^*$ be such that $x \notin \lang$. We claim that  
$\prob\left[ \accept \leftarrow \langle P^*(x),V(x) \rangle \right] \leq \negl(\secparam)$.
\par From the soundness of $\protwi$, the following holds:
\begin{eqnarray}
\label{eqn:soundness}
\prob\left[ \accept \leftarrow \langle P^*(x),V(x) \rangle \right] \leq \prob\left[ \bfx_{\sf wi} \in \lang_{\sf wi} \right] + \negl(\secparam),
\end{eqnarray}
where $x_{\sf wi}$ is a random variable that is assigned the WI instance sampled in the execution of $P^*(x)$. 
\par Therefore, to prove soundness, it suffices upper bound the value $\eps := \prob[\bfx_{\sf wi} \in \lang_{\sf wi}]$.

\begin{lemma}
\label{clm:wisoundness}
$\eps \leq \negl(\secparam)$.
\end{lemma}
\begin{proof}
Let $\bfx_{\sf wi}= \left( x,\bfc^*,\bfc^{**},\left\{\alpha_{i,b} \right\}_{i \in [\secparam],b\in \{0,1\}}\right)$ be the $\protwi$ instance computed during the execution of $P^*(x)$ and $V(x)$. Define the following event, parameterized by a set of values: \\

\noindent Event ${\bf E}\left[ \left\{ \alpha_{i,b} \right\}_{i \in [\secparam],b \in\{0,1\}} \right]$: Output 1 if there exists an ${\bf i} \in [\secparam]$ such that the commitment ${\bf c}^{**}$ sent by $P^*$ is a commitment of $(\alpha_{{\bf i},0} \oplus \alpha_{{\bf i},1})$. \\

\noindent Consider the following hybrid argument. \\

\noindent \underline{$\hybrid_0$}: This corresponds to the execution of $P^*(x)$ and $V(x)$.
\par Let $\{\alpha_{i,b}\}_{i \in [\secparam],b \in \{0,1\}}$ be the values committed  by $V$ in Step 2. Since $x \notin \lang$, the following holds: 
$$\prob\left[ {\bf E}\left[ \left\{ \alpha_{i,b} \right\}_{i \in [\secparam],b \in\{0,1\}} \right] = 1 \text{ in }\hybrid_0 \right] = \prob \left[ \bfx_{\sf wi} \in \lang_{\sf wi} \right] = \eps$$

\noindent \underline{$\hybrid_1$}: We modify the description of $V(x)$ as follows. Sample $\alpha_i \xleftarrow{\$} \{0,1\}^{\secparam},\alpha_{i,b} \xleftarrow{\$} \{0,1\}^{\secparam}$ such that $\alpha_i = \alpha_{i,0} \oplus \alpha_{i,1}$. 
\par The two hybrids $\hybrid_0$ and $\hybrid_1$ are identical. Thus, we have the following:
$$\prob\left[ {\bf E}\left[ \left\{ \alpha_{i,b} \right\}_{i \in [\secparam],b \in\{0,1\}} \right] = 1 \text{ in }\hybrid_1 \right] = \eps$$

\noindent \underline{$\hybrid_{2.(i^*,b^*)}$ for $i^*\in [2\secparam],b^* \in \{0,1\}$}: The verifier behaves as in $\hybrid_1$ except for the following: for every $i \in [\secparam]$, it additionally samples $\beta_{i,0} \xleftarrow{\$} \{0,1\}^{\secparam},\beta_{i,1} \xleftarrow{\$} \{0,1\}^{\secparam}$. Furthermore, for every $i \leq i^*$, $b \leq b^*$, the $(i,b)^{th}$ committed value by $V$ in Step 2 is $\beta_{i,b}$. We emphasize that for every $i \in [\secparam]$, it still samples $\alpha_{i,0} \xleftarrow{\$}\{0,1\}^{\secparam}, \alpha_{i,1} \xleftarrow{\$} \{0,1\}^{\secparam}$.
\par We prove the following two claims.  

\begin{claim}
\label{clm:soundness:first}
Assuming the quantum-concealing property of $(\comm,\receiver,\open)$, there exists a negligible function $\nu(\secparam)$ such that the following holds:
\begin{enumerate}
    \item For every $i^* \in [\secparam]$,
    $$\left| \prob\left[ {\bf E}\left[ \left\{ \alpha_{i,b} \right\}_{i \in [\secparam],b \in\{0,1\}} \right] \text{ in } \hybrid_{2.(i^*.0)}  \right] -  \prob\left[ {\bf E}\left[ \left\{ \alpha_{i,b} \right\}_{i \in [\secparam],b \in\{0,1\}} \right] \text{ in } \hybrid_{2.(i^*.1)} \right] \right| \leq \nu(\secparam)$$
    \item For every $i^* \in [\secparam-1]$,
    $$\left| \prob\left[ {\bf E}\left[ \left\{ \alpha_{i,b} \right\}_{i \in [\secparam],b \in\{0,1\}} \right] \text{ in } \hybrid_{2.(i^*-1.1)}  \right] -  \prob\left[ {\bf E}\left[ \left\{ \alpha_{i,b} \right\}_{i \in [\secparam],b \in\{0,1\}} \right] \text{ in } \hybrid_{2.(i^*.0)} \right] \right| \leq \nu(\secparam),$$
    \item
    $$\left| \prob\left[ {\bf E}\left[ \left\{ \alpha_{i,b} \right\}_{i \in [\secparam],b \in\{0,1\}} \right] \text{ in } \hybrid_{1}  \right] -  \prob\left[ {\bf E}\left[ \left\{ \alpha_{i,b} \right\}_{i \in [\secparam],b \in\{0,1\}} \right] \text{ in } \hybrid_{2.(1.0)} \right] \right| \leq \nu(\secparam)$$
\end{enumerate}
\end{claim}
\begin{proof}
We will prove just the first bullet and the other items follow analogously. 
\par Let $\delta_{i.{\bf b}}=\prob\left[ {\bf E}\left[ \left\{ \alpha_{i,b} \right\}_{i \in [\secparam],b \in\{0,1\}} \right] \text{ in } \hybrid_{2.(i.{\bf b})}  \right]$. Suppose the first bullet does not hold. This means that there exits $i^* \in [\secparam]$ and a non-negligible function $\nu(\secparam)$ such that $|\delta_{i^*.0} - \delta_{i^*.1}| \geq \nu(\secparam)$. 
\par We construct a non-uniform reduction ${\cal B}$ as follows. It receives as non-uniform advice, $(\bfr^*,\bfc^*,\sigma_{\sf adv})$ generated as follows: 
\begin{itemize}[noitemsep]
\item Run $P^*(x)$ to generate $\bfc^*$.
\item Let $\bfr^*$ be the value committed in $\bfc^*$. 
\item Let $\sigma_{\sf adv}$ be the state of $P^*(x)$ after the execution of step 1 of $\Pi$ in~\Cref{fig:pcoinczk}. 
\end{itemize}
\noindent ${\cal B}$ proceeds as follows: \begin{itemize}
    \item For every $i\in \{0,1\}^{\secparam}$, it samples $\alpha_{i,0},\alpha_{i,1},\beta_{i,0},\beta_{i,1} \xleftarrow{\$} \{0,1\}^{\secparam}$,  
    \item It runs $P^*(x)$ starting from step 2 of $\Pi$ in~\Cref{fig:pcoinczk}. $P^*$ is initialized with the state $\sigma_{{\sf adv}}$. 
    \item For all the $(i,0)^{th}$ and $(i,1)^{th}$ executions of step 2 of $\Pi$, ${\cal B}$ plays the role of $V(x)$. That is, in the $(i,b)^{th}$ execution, for $i < i^*$, it commits to $\beta_{i,b}$. In the $(i,b)^{th}$ execution for $i > i^*$, it commits to $\alpha_{i,b}$. Similarly, for $i=i^*$, in the  $(i,0)^{th}$ execution, it commits to $\beta_{i,b}$. For the $(i^*,1)^{th}$ execution, it acts as an intermediary between the external challenger, which plays the role of sender in $(\comm,\receiver,\open)$, and the prover $P^*(x)$. The challenger uses the challenge message pair  $(\alpha_{i^*,1},\beta_{i^*,1})$ in the quantum-concealing experiment when interacting with ${\cal B}$, who plays the role of the malicious receiver.  
    
    \item ${\cal B}$ continues the execution of $P^*(x)$ by emulating the verifier $V(x)$ until step 5 of $\Pi$. 
    
    \item Let $\bfc^{**}$ be the commitment sent by $P^*(x)$ in step 5. 
    
    \item ${\cal B}$ recovers the message $\gamma$ from $\bfc^{**}$ using ${\bfr^*}$. If the recovery fails, it aborts. 
    
    \item If there exists $i \in [\secparam]$ such that $\gamma = \alpha_{i,0} \oplus \alpha_{i,1}$, ${\cal B}$ outputs 1. Otherwise, it outputs 0. 
    
\end{itemize}
\noindent If the challenger uses the message $\alpha_{i^*,1}$ in the execution of the commitment protocol then the probability that ${\cal B}$ outputs 1 is precisely $\delta_{i^*,0}$. Similarly, if the challenger uses $\beta_{i^*,1}$ then the probability that ${\cal B}$ outputs 1 is precisely $\delta_{i^*,1}$. 
\par Since $|\delta_{i^*,0} - \delta_{i^*,1}| \geq \nu(\secparam)$, where $\nu(\secparam)$ is non-negligible, we have that ${\cal B}$ violates the quantum-concealing property of $(\comm,\receiver,\open)$, a contradiction. Thus, $\nu(\secparam)$ is negligible which completes the proof of the claim. 
\end{proof}

\noindent As a corollary and from the fact that $\prob\left[ {\bf E}\left[ \left\{ \alpha_{i,b} \right\}_{i \in [\secparam],b \in\{0,1\}} \right] = 1 \text{ in }\hybrid_1 \right] = \eps$, we have the following. 

\begin{corollary}
\label{cor:soundness}
$$\left| \prob\left[ {\bf E}\left[ \left\{ \alpha_{i,b} \right\}_{i \in [\secparam],b \in\{0,1\}} \right] \text{ in } \hybrid_{2.(\secparam.1)} \right] \right| \geq \eps - (2\secparam) \cdot  \nu(\secparam)$$ 
\end{corollary}
\begin{proof}
Follows by applying triangle inequality on~\Cref{clm:soundness:first}.
\end{proof}

\begin{claim}
\label{clm:soundness:second}
$\prob\left[ {\bf E}\left[ \left\{ \alpha_{i,b} \right\}_{i \in [\secparam],b \in\{0,1\}} \right] = 1 \text{ in }\hybrid_{2.(\secparam.1)} \right] \leq \frac{\secparam}{2^{\secparam}}$. 
\end{claim}
\begin{proof}
In $\hybrid_{2.(\secparam.1)}$, $\{\alpha_{i,b}\}_{i \in [\secparam],b \in \{0,1\}}$ is information-theoretically hidden from $P^*(x)$. Thus, for any given $i\in[\secparam]$, the probability that $\alpha_{i,0} \oplus \alpha_{i,1}$ is committed by $P^*$ in bullet 5 of $\Pi$ is $\frac{1}{2^{\secparam}}$. By union bound, the proof of the claim follows. 
\end{proof}

\noindent From~\Cref{cor:soundness} and~\Cref{clm:soundness:second}, $\eps$ is a negligible function in $\secparam$.

\end{proof}

\noindent From~\Cref{clm:wisoundness} and Equation~\ref{eqn:soundness}, it follows that $\prob\left[ \accept \leftarrow \langle P^*(x),V(x) \rangle \right] \leq \negl(\secparam)$. 

\end{proof}

\subsection{Proof of Quantum Zero-Knowledge}
\begin{proposition}
Let $M = O(\log \lambda)$ and $\functionspace(M) = 2 M$.
The protocol described in \Cref{fig:pcoinczk} is $(M, \functionspace)$-space bounded QZK.
\end{proposition}
\begin{proof}
Suppose $V^*$ be a malicious QPT (non-uniform) verifier with $M$-qubit quantum space. We describe a QPT simulator $\Sim$ such that the output distribution of $\Sim$ is computationally indistinguishable from the output distribution of the real world. 

\paragraph{Description of Simulator.} Let ${\bf R}$ and ${\bf X}$ be two $M$-qubit registers. The register ${\bf X}$ is $V^*$'s private register whereas ${\bf R}$ is an auxiliary register maintained by the simulator. The verifier $V^*$'s advice state is stored in the register ${\bf X}$. We describe the simulator $\Sim$  in~\Cref{fig:sim}.

\begin{figure}[!htb]
   \begin{center} 
   \begin{tabular}{|p{16cm}|}
   \hline
\begin{enumerate}
    \item By taking the role of the honest prover, run the first four steps in~\Cref{fig:pcoinczk}. That is, until the simulator receives $\{(\alpha_{i,\bfb_i},r_{i,\bfb_i})\}_{i \in [\secparam]}$. Abort if $V^*$ aborts at any point. We emphasize that $V^*$ is executed on input the register ${\bf X}$. Denote $\rho_{{\sf res}}$ to be the quantum residual state after this step. Denote the classical state to be ${\sf st}$.
    \label{step:first-run}
    \begin{enumerate}
        \item $\Sim$ and $V^*$ engage in an execution of the commitment protocol, where $\Sim$ commits to $\bfr^*$, which is sampled as $\bfr^* \xleftarrow{\$} \{0,1\}^{\secparam}$. Denote the commitment to be $\bfc^*$. 
       
        \item $\Sim$ and $V^*$ run $2\secparam$ number of executions of the commitment protocol. Denote the commitments in all the $2\secparam$ executions by $(\{\bfc'_{i,b}\}_{i \in [\secparam],b\in \{0,1\}})$.

        \item $\Sim$ sends $\bfb_1,\ldots,\bfb_{\secparam}$, where $\bfb_i \xleftarrow{\$} \{0,1\}^{\secparam}$, to $V^*$. 
        
        \item $\Sim$ receives the openings of all the commitments $(\{\bfc_{i,\bfb_i}\}_{i \in [\secparam]})$. Denote the opening of the $(i,\bfb_i)^{th}$ commitment to be $o_{i,\bfb_i}=\open(\bfc_{i,\bfb_i})$. If any of the openings are invalid, reject. Otherwise, parse $o_{i,\bfb_i}$ as $(\alpha_{i,\bfb_i},r_{i,\bfb_i})$.  
    \end{enumerate}
    \item Repeat the following forever: 
    \begin{enumerate}
        \item Initialize the register ${\bf R}$ with the maximally mixed state $\frac{I}{2^M}$.
        \item Run the third and fourth steps of~\Cref{fig:pcoinczk}, where $V^*$ is executed on the register ${\bf R}$. That is, send $\bfb'_1,\ldots,\bfb'_{\secparam} \xleftarrow{\$} \{0,1\}^{\secparam}$. Receive the openings from $V^*$, denoted by $\{(\alpha_{i,\bfb'_{i}},r_{i,\bfb'_i})\}_{i \in [\secparam]}$. Abort the loop if $V^*$ did not abort and $\exists i \in [\secparam]$ such that $\bfb_i \neq \bfb'_i$. Otherwise, continue.   \label{step:iteration}
    \end{enumerate}
    \item Discard the register ${\bf R}$. From here on, continue (from Step~\ref{step:first-run}) the execution of $V^*$ on the register ${\bf X}$. Specifically, run $V^*$ on $(\rho_{\sf res},{\sf st})$.  
    \item The simulator commits to $(i^*,\alpha_{i^*,0} \oplus \alpha_{i^*,1})$, where $i^* \in [\secparam]$ is such that $\bfb_{i^*} \neq \bfb'_{i^*}$. Denote the commitment to be $\bfc^{**}$ and let $\bfr^*$ be the randomness used in $\bfc^{**}$. 
    \item It runs the last step of the trapdoor phase by emulating the execution of the honest prover.
    \item Finally, $\Sim$ runs the WI phase with $V^*$. Let the instance be $\bfx_{\sf wi}= \left( x,\bfc^*,\bfc^{**},\left\{\alpha_{i,b} \right\}_{i \in [\secparam],b\in \{0,1\}}\right)$, where $\bfc^*$ be the commitment sent by $\Sim$ in the beginning of the protocol. $\Sim$ emulates the honest prover except the witness used is $(\bot,\bfr^*,r,i^*)$, where $r$ is the randomness used in the commitment of $\bfc^*$.   
\end{enumerate} 
~\\
    \hline
   \end{tabular}
    \caption{(Expected) QPT Simulator $\Sim$.}
    \label{fig:sim}
    \end{center}
\end{figure}

\par We need to show that $\Sim$ indeed runs in expected QPT. We also need to demonstrate that the output distribution of $\Sim$ is computationally indistinguishable from the real world. 
\par We first prove a simple claim. Define $\Lambda$ be a POVM that on input $\sigma$, defined on ${\bf X}$, implements steps 3 and 4 of $\Pi$ in~\Cref{fig:pcoinczk} and outputs 1 if and only if $V^*$ does {\bf not} abort. We assume that the classical state is hardcoded in the description of $\Lambda$. Define another POVM $\Gamma$ that on input $\sigma$, defined on ${\bf R}$, implements 2.2 of~\Cref{fig:sim} and outputs 1 if and only if $V^*$ did {\bf not} abort and $\exists i \in [\secparam]$ such that $\bfb_i \neq \bfb'_i$. 

\begin{claim}
\label{clm:auxclm:sim}
Suppose $p = \Tr(\Lambda(\rho_{\sf res}))$, where $\rho_{\sf res}$ is as defined in \Cref{fig:sim}. Then, $ \Tr(\Gamma(\frac{I}{2^M})) \geq \frac{p}{2^M} \cdot \left( 1 - 2^{-\secparam}\right)$.
\end{claim}
\begin{proof}
Firstly, note that $\Tr(\Lambda(\sigma)) \cdot \left(1 - {2^{-\secparam}} \right) = \Tr(\Gamma(\sigma))$, for every density matrix $\sigma$. This follows from the fact that the probability that $\exists i \in [\secparam]$ such that $\bfb_i \neq \bfb'_i$ is precisely $(1 - 2^{-\secparam})$. 
\par Now, we estimate $\Tr(\Lambda(\frac{I}{2^M}))$. We have:
\begin{eqnarray*}
\Tr\left( \Lambda \left( \frac{I}{2^M} \right) \right) & = & \frac{1}{2^M} \Tr\left( \Lambda \left( I \right) \right) \\
& = &  \frac{1}{2^M} \Tr\left( \Lambda \left( \rho_{\sf res} + I - \rho_{\sf res} \right) \right) \\
& \geq & \frac{1}{2^M} \Tr\left( \Lambda \left( \rho_{\sf res} \right) \right)\ \ (\because\ I - \rho_{\sf res} \succeq 0) \\
& = & \frac{p}{2^M}
\end{eqnarray*}
Thus, $\Tr \left( \Gamma(\frac{I}{2^M}) \right) = \Tr \left( \Lambda(\frac{I}{2^M}) \right) \cdot \left( 1 - 2^{-\secparam} \right) \geq \frac{p}{2^M} \cdot (1 - 2^{-\secparam})$. 
\end{proof}

\paragraph{Runtime Analysis.}\ Let $q(\secparam)$ be an upper bound on the runtime of all the steps (1 through 6) in~\Cref{fig:sim}. Let $\rho$ be the auxiliary state of $V^*$. We use the notation $\Lambda$ and $\Gamma$ as defined earlier. Let $p=\Tr(\Lambda(\rho_{res}))$ and  $p'=\Tr(\Gamma(\frac{I}{2^M}))$. From~\ref{clm:auxclm:sim}, $p' \geq \frac{p}{2^M}(1 - 2^{-\secparam})$. There are three cases. \\

\noindent {\em Case 1: $p=0$.} In this case, $\Sim$ aborts after step 1 of~\Cref{fig:sim}. Thus, $\Sim$ runs in polynomial time. \\

\noindent {\em Case 2: $p'=1$.} In this case, either $\Sim$ aborts in step 1 or it runs the loop in step 2 only once. In either of the two cases, $\Sim$ runs in polynomial time. \\
 
\noindent {\em Case 3: $p > 0$ and $p' < 1$.}\ \\
\begin{eqnarray*}
\text{Expected Runtime}
 & = & q(\secparam) \left( (1-p) \cdot 1 +  p \sum_{i=1}^{\infty} i \cdot (1-p')^{i-1} \cdot p'  \right) + 4q(\secparam) \\
  & = & q(\secparam) \left( (1-p) \cdot 1 +  \frac{p \cdot p'}{(1-p')} \sum_{i=1}^{\infty} i \cdot (1-p')^{i}  \right) + 4q(\secparam) \\
  & = & q(\secparam) \left( (1-p) \cdot 1 +  \frac{p \cdot p'}{(1-p')} \cdot \frac{(1 - p')}{(p')^2} \right) + 4q(\secparam)\ \left(\because \sum_{i=1}^{\infty} i \cdot c^i = \frac{c}{(1-c)^2} \text{ for } c > 0 \right) \\
    & = & q(\secparam) \left( (1-p) \cdot 1 +  \frac{p}{p'} \right) + 4q(\secparam) \\
& \leq & q(\secparam)\left( (1-p) + \frac{p}{\frac{p(1-2^{-\secparam})}{2^M}} \right) + 4q(\secparam) \\
& \leq &  q(\secparam)\left( (1-p) + 2^{M+1} \right) + 4q(\secparam)\\
& = & \poly(\secparam)
\end{eqnarray*}
The last equality follows from the fact that $M$ is logarithmic in $\secparam$. 

\paragraph{Indistinguishability of Real and Ideal worlds.} Let $(x,w) \in \rel(\lang)$. Let $\rho$ be the auxiliary state of $V^*$. 

\begin{lemma}
 $\left\{ \view_{V^*}\left(\langle P(x,w),V^*(x,\rho) \right) \right\} \approx_{\quant} \left\{ {\sf Sim}(x,\rho) \right\}$. 
\end{lemma}
\begin{proof}
Consider the following hybrid argument. \\

\noindent \underline{$\hybrid_0$}: The output distribution of this hybrid is $\view_{V^*}\left(\langle P(x,w),V^*(x,\rho) \right)$. \\

\noindent \underline{$\hybrid_1$}: We define a hybrid simulator that essentially behaves like an honest prover except it keeps rewinding until it finds two accepting transcripts. We emphasize that, just like the honest prover, the hybrid simulator still uses the witness $w$ in the WI phase. Formally, we consider a hybrid simulator $\hybrid_1.\Sim(x,w)$ defined as follows: 
\begin{enumerate}
    \item Execute steps 1 through 3 of $\Sim$ in~\Cref{fig:sim}. 
    \begin{enumerate}
        \item Run the following steps:
        \begin{enumerate}
        \item Engage in an execution of the commitment protocol with $V^*$, by committing to $\bfr^*$, which is sampled as $\bfr^* \xleftarrow{\$} \{0,1\}^{\secparam}$. Denote the commitment to be $\bfc^*$. 
       
        \item Execute $2\secparam$ number of executions of the commitment protocol with $V^*$, with $V^*$ playing the role of the committer. Denote the commitments in all the $2\secparam$ executions by $(\{\bfc'_{i,b}\}_{i \in [\secparam],b\in \{0,1\}})$.

        \item Send $\bfb_1,\ldots,\bfb_{\secparam}$, where $\bfb_i \xleftarrow{\$} \{0,1\}^{\secparam}$, to $V^*$. 
        
        \item Receive the openings of all the commitments $(\{\bfc_{i,\bfb_i}\}_{i \in [\secparam]})$. Denote the opening of the $(i,\bfb_i)^{th}$ commitment to be $o_{i,\bfb_i}=\open(\bfc_{i,\bfb_i})$. If any of the openings are invalid, reject. Otherwise, parse $o_{i,\bfb_i}$ as $(\alpha_{i,\bfb_i},r_{i,\bfb_i})$.  
    \end{enumerate}
    \item Repeat the following forever: 
    \begin{enumerate}
        \item Initialize the register ${\bf R}$ with the maximally mixed state $\frac{I}{2^M}$.
        \item Run the third and fourth steps of~\Cref{fig:pcoinczk}, where $V^*$ is executed on the register ${\bf R}$. That is, send $\bfb'_1,\ldots,\bfb'_{\secparam} \xleftarrow{\$} \{0,1\}^{\secparam}$. Receive the openings from $V^*$, denoted by $\{(\alpha_{i,\bfb'_{i}},r_{i,\bfb'_i})\}_{i \in [\secparam]}$. Abort the loop if $V^*$ did not abort and $\exists i \in [\secparam]$ such that $\bfb_i \neq \bfb'_i$. Otherwise, continue
    \end{enumerate}
    \item Discard the register ${\bf R}$. From here on, continue (from Step~\ref{step:first-run}) the execution of $V^*$ on the register ${\bf X}$.  
    \end{enumerate}
    \item Execute steps 5 and 6 of $\Pi$ in~\Cref{fig:pcoinczk}. Also, execute the WI phase. 
    \begin{enumerate}
    \item Engage in another execution of the commitment protocol by committing to $0$, using the randomness $\bfr^*$. Denote the commitment to be $\bfc^{**}$. 
        
        \item $V^*$ output the rest of the openings of $\left( \{\bfc_{i,b}\}_{i \in [\secparam],b\in \{0,1\}} \right)$. That is, it sends $\{o_{i,1-\bfb_i}\}_{i \in [\secparam]}$, where $o_{i,\bfb_i}=\open(\bfc_{i,1-\bfb_i})$. If any of the openings are invalid, $P$ rejects. Otherwise, $P$ parses $o_{i,1-\bfb_i}$ as $(\alpha_{i,1-\bfb_i},r_{i,1-\bfb_i})$. 
        \item Engage in $\protwi$ with $V^*$ on the following common input: $$\bfx_{\sf wi}= \left( x,\bfc^*,\bfc^{**},\left\{\alpha_{i,b} \right\}_{i \in [\secparam],b\in \{0,1\}}\right)$$ Additionally, use the witness  $(w,\bot,\bot,\bot)$ in $\protwi$.
    \end{enumerate}
\end{enumerate}
The output of $\hybrid_1.\Sim(x,w)$ is the view of $V^*$ in $\hybrid_1$.  
\par Consider the following claim. 

\begin{claim}
The trace distance between $\hybrid_0$ and $\hybrid_1.\Sim(x,w)$ is at most $\negl(\secparam)$.
\end{claim}
\begin{proof}
Note that if $\hybrid_1.\Sim(x,w)$ terminates then the density matrices output in $\hybrid_0$ and $\hybrid_1.\Sim(x,w)$ are identically distributed. It suffices to upper bound the probability that $\hybrid_1.\Sim(x,w)$ does not terminate. Let $p=\Tr(\Lambda(\rho))$, where $\rho$ is the advice state of $V^*$. Let $p'=\Tr(\Gamma(\rho))$. From~\Cref{clm:auxclm:sim}, $p' \geq \frac{p(1 - 2^{-\secparam})}{2^M}$. There are two cases: 
\begin{itemize}
    \item {\em Case 1. $p=0$ or $p'=1$}: As we argued in the proof of~\Cref{clm:auxclm:sim}, $\Sim$ terminates in this case.  
    \item {\em Case 2. $p \neq 0$ and $p' \neq 1$}: In this case, we claim that the simulator terminates within $\secparam \cdot (p')^{-1}$ iterations with overwhelming probability. We show this below. 
    \begin{eqnarray*}
 \prob\left[\Sim\ {\text{ does not terminate within }}\secparam \cdot (p')^{-1}{\text{ iterations}}\right] 
   \leq (1 - p')^{\secparam \cdot (p')^{-1}} 
   = e^{-\secparam}
    \end{eqnarray*}
    
\end{itemize}
\end{proof}

\noindent \underline{$\hybrid_2$}: We define another hybrid simulator $\hybrid_2.\Sim$ which behaves like $\hybrid_1.\Sim$, except in the following steps: 
\begin{itemize}
    \item[1.1.i] Engage in an execution of the commitment protocol with $V^*$, by committing to $0$. Denote the commitment to be $\bfc^*$.\\
    {\em Note: in bullet 2.1  of $\hybrid_1$, $\bfc^{**}$ is still a commitment computed using the randomness $\bfr^*$, where $\bfr^* \xleftarrow{\$} \{0,1\}^{\secparam}$.}
\end{itemize}
The rest of the steps is the same as $\hybrid_1.\Sim$. The output of $\hybrid_2.\Sim$ is the view of $V^*$ in $\hybrid_2$. 
\par The output distributions of $\hybrid_1.\Sim$ and $\hybrid_2.\Sim$ are computationally indistinguishable from the quantum-concealing property of $(\comm,\open,\receiver)$.\\

\noindent \underline{$\hybrid_3$}: We define a hybrid simulator $\hybrid_3.\Sim$ which behaves like $\hybrid_2.\Sim$, except in the following steps: 
\begin{itemize}
    \item[2.1] Engage in an execution of the commitment protocol with $V^*$ by committing to $(i^*,\alpha_{i^*,0} \oplus \alpha_{i^*,1})$, where $i^* \in [\secparam]$ is such that $\bfb_{i^*} \neq \bfb'_{i^*}$. Denote the commitment to be $\bfc^{**}$.
\end{itemize}
\noindent The rest of the steps is the same as $\hybrid_2.\Sim$. The output of $\hybrid_3.\Sim$ is the view of $V^*$ in $\hybrid_3$. 
\par The output distributions of $\hybrid_2.\Sim$ and $\hybrid_3.\Sim$ are computationally indistinguishable from the quantum-concealing property of $(\comm,\open,\receiver)$.\\

\noindent \underline{$\hybrid_4$}: We define a hybrid simulator $\hybrid_4.\Sim$ which behaves like $\hybrid_3.\Sim$, except in the following steps: 
\begin{itemize}
    \item[1.1.i] Engage in an execution of the commitment protocol with $V^*$, by committing to $\bfr^*$, which is sampled as $\bfr^* \xleftarrow{\$} \{0,1\}^{\secparam}$. Denote the commitment to be $\bfc^*$. 
\end{itemize}
\noindent The rest of the steps is the same as $\hybrid_3.\Sim$. The output of $\hybrid_4.\Sim$ is the view of $V^*$ in $\hybrid_4$. 
\par The output distributions of $\hybrid_3.\Sim$ and $\hybrid_4.\Sim$ are computationally indistinguishable from the quantum-concealing property of $(\comm,\open,\receiver)$.\\

\noindent \underline{$\hybrid_5$}: The output of this hybrid is the output of $\Sim(x,\rho)$. 
\par The only difference between $\hybrid_4.\Sim$ and $\Sim$ is in the WI phase. The computational indistinguishability of $\hybrid_4$ and $\hybrid_5$ follows from the witness indistinguishability property of $\protwi$. 
\end{proof}

\end{proof}

\section{$\omega(\log(\secparam))$-Quantum Space Verifiers: Impossibility Result}
\label{sec:impossibility} 

\noindent In the previous section, we showed the existence of quantum zero-knowledge protocols secure against logarithmic quantum space verifiers. We present a complementary result. We show that there does not exist any quantum zero-knowledge protocol with fully black box simulation for languages outside BQP against super-logarithmic quantum space verifiers. 

\begin{lemma}\label{lem:negative-result}
Let $s(\secparam)=\omega(\log(\secparam))$ and $\functionspace$ be any function such that $\functionspace(x) > 2x$. Assume the existence of post-quantum one-way functions. Suppose there exists a fully black-box $(s,\functionspace)$-space bounded quantum zero-knowledge protocol for a language ${\cal L}$. Then, ${\cal L}$ is in {\sf BQP}. 
\end{lemma}

\newcommand{\dec}{\mathsf{Dec}}
\newcommand{\gen}{\mathsf{Gen}}
\begin{proof}
Consider a language ${\cal L}$. Suppose there exists a $k$-round 
fully black-box $(s,\functionspace)$-space bounded quantum zero-knowledge protocol ${\cal L}$ for some $s(\secparam)=\omega(\log(\secparam))$ and for some function $\functionspace$.
\par We describe a malicious verifier $V'$ below. In addition to subspace states (\Cref{sec:querylb}), we use the following additional tools:
\begin{itemize}
    \item Symmetric-key encryption with post-quantum CPA (chosen plaintext attack) security, denoted by $(\setup,\enc,\dec)$. 
    \item Post-quantum secure digital signatures satisfying existential unforgeability property, denoted by $(\gen,\sign,\verify)$. 
\end{itemize}
\noindent Both of these primitives can be built from post-quantum secure one-way functions~\cite{Goldreich2001,NY89}. 

\paragraph{Description of $V'$.} Consider the following QPT verifier $V'$ such that the quantum memory of $V'$ is $M=\omega(\secparam) $ qubits. Hardwired in its code are $k$ random $\frac{M}{2}$-dimensional subspaces $S_0$,...,$S_k$, where $S_i \subseteq \mathbb{F}_2^{M}$, a secret-key $sk_{\text{enc}}$ for a symmetric-key encryption scheme, sampled as $sk_{\text{enc}} \leftarrow \setup(1^{\secparam})$, and a secret-key and verification keys $(sk_{\text{sign}},vk_{\text{sign}})$ for a digital signature scheme, sampled as $(sk_{\text{sign}},vk_{\text{sign}}) \leftarrow \gen(1^{\secparam})$. 
  $V'$
  receives as input $M$-qubit subspace state $\ket{S_0}$, which is initialized in the register ${\bf Y}$. %
 \\ 
  
\noindent   In the first round of the protocol, $V'$ gets as input two registers ${\bf X}$, ${\bf Y}$ and does the following:  
\begin{enumerate}
    \item Measure the register ${\bf X}$ to obtain $\alpha_1$. 
    \item Check if the state held in ${\bf Y}$ is $\ketbra{S_0}{S_0}$. That is, perform the projective measurement $\{\underbrace{\kb{S_0}}_{\text{outcome}=0}, \underbrace{I-\kb{S_0}}_{\text{outcome}=1}\}$ on the register ${\bf Y}$. This is possible since $V'$ has the description of $S_0$ hardwired in its code.
    \item If the measured outcome is 1, abort. 
    \item Run the first round of $V$ on $\alpha_1$ to obtain $\beta_1$. Abort if $V$ aborts. 
    \item Re-initialize the register ${\bf X}$ with $\ketbra{\beta_1}{\beta_1}$.
    \item Re-initialize the register ${\bf Y}$ with a  new $M$-qubit subspace state $\ketbra{S_{1}}{S_{1}}$ in the memory.
    \item Using the hardcoded $sk_{\text{enc}}$, initialize the register ${\bf Z}$ with the $c_1 = \enc(sk_{\text{enc}},\gamma_1)$, where $\gamma_1$ is the state of $V$ after the answer.
    \item Using the hardcoded $sk_{\text{sign}}$, initialize ${\bf T}$ with the signature of $c_1$, i.e., $\sign(sk_{\text{sign}},c_1)$. 
    \item Output the registers ${\bf X}$,${\bf Y}$,${\bf Z}$ and ${\bf T}$. 
\end{enumerate}

  \noindent In the $i^{th}$ round of the protocol, $V'$ gets as input four registers ${\bf X}$, ${\bf Y}$, ${\bf Z}$ and ${\bf T}$ and does the following:  
\begin{enumerate}
    \item Measure the register ${\bf X}$ to obtain $\alpha_i$. 
    \item Check if the state held in ${\bf Y}$ is $\ketbra{S_{i-1}}{S_{i-1}}$. That is, perform the projective measurement $\{\underbrace{\kb{S_{i-1}}}_{\text{outcome}=0}, \underbrace{I-\kb{S_{i-1}}}_{\text{outcome}=1}\}$ on the register ${\bf Y}$. This is possible since $V'$ has the description of $S_i$ hardwired in its code.
    \item If the measured outcome is 1, abort. 
    \item Using the hardcoded $vk_{\text{sign}}$, check if ${\bf T}$ contains the signature of the value in register ${\bf Z}$. If the verification does not pass, abort.
    \item Using the hardcoded $sk_{\text{enc}}$, decrypt ${\bf Z}$ leading to value $\gamma_{i-1}$
    \item Run the $i^{th}$ round of $V$ on $\alpha_i$ with internal state $\gamma_{i-1}$ a to obtain $\beta_i$. Abort if $V$ aborts. 
    \item Re-initialize register ${\bf X}$ with $\ketbra{\beta_i}{\beta_i}$.
    \item Re-initialize the register ${\bf Y}$ with a  new $M$-qubit subspace state $\ketbra{S_{i}}{S_{i}}$ in the memory.
        \item Using the hardcoded $sk_{\text{enc}}$, re-initialize the register ${\bf Z}$ with the $c_i = \enc(sk_{\text{enc}},\gamma_i)$, where $\gamma_i$ is the state of $V$ after the answer.
    \item Using the hardcoded $sk_{\text{sign}}$, initialize ${\bf T}$ with the signature of $c_i$, i.e., , $\sign(sk_{\text{sign}},c_i)$. 
    \item Output the registers ${\bf X}$,${\bf Y}$,${\bf Z}$ and ${\bf T}$. 
\end{enumerate}

\noindent Note that the quantum memory of $V'$ takes $M$ qubits.
\par Let $\channel(V') = (\Phi_1,\cdots,\Phi_k)$  be the quantum channels corresponding to the operations of $V'$ as defined in~\Cref{sec:fully-black-box}. Since $\Pi$ satisfies fully black-box  $(M,\functionspace)$-space bounded QZK, there exists a QPT simulator
 $\Sim^{\channel(V')}$, taking space at most $p
 (M)$ qubits, simulating the interaction of $V'$ with the honest prover $P$. In particular,  $\Sim^{\channel(V')}$ does not abort with probability $1-\negl(\lambda)$, in which case it outputs $\ket{S_{k}}$.

\paragraph{Helpful definitions.} We state some useful definitions below. 

\begin{definition}[State-successful queries]
A query $\rho_{{\bf X} \otimes {\bf Y} \otimes {\bf Z} \otimes {\bf T}}$ made by $\Sim$ to $\Phi_i$, for some $i \in [k]$, is said to be 
{\em state-successful} if $\Phi_i(\rho_{{\bf X} \otimes {\bf Y} \otimes {\bf Z} \otimes {\bf T}})$ does not abort on step $3$. 
\end{definition}

\noindent In particular, a state-successful query made by $\Sim$ could lead to abort if the other registers ${\bf X},{\bf Z}$ and ${\bf T}$ are not of a specific form. 

\begin{definition}[(Non-)abort queries]
\label{def:nonabort:queries}
A query $\rho_{{\bf X} \otimes {\bf Y} \otimes {\bf Z} \otimes {\bf T}}$ made by $\Sim$ to $\Phi_i$, for some $i \in [k]$, is said to be 
{\em non-abort} if $\Phi_i(\rho_{{\bf X} \otimes {\bf Y} \otimes {\bf Z} \otimes {\bf T}})$  does not abort. Similarly, we say that it is an abort query if  $\Phi_i(\rho_{{\bf X} \otimes {\bf Y} \otimes  {\bf Z} \otimes {\bf T}})$  aborts.
\end{definition}

\noindent Note that a non-abort query is also a state-successful query. 
 
\begin{definition}[Successful simulation]
We say that $\Sim$ is {\em successful}, if it does not abort.
\end{definition} 

\paragraph{Helpful Claim.} We now state an important claim regarding the properties of queries made by $\Sim$ to $\Phi_1,...,\Phi_k$.  
  
\begin{claim}
\label{lem:structure-simulator}
Except with $\negl(\lambda)$ probability, the following events do not occur:
\begin{enumerate}
    \item for some $i$, $\Sim$ has a state-successful query to $\Phi_{i}$ without having a non-abort query to $\Phi_{i-1}$;
    \item for some $i \leq j$, $\Sim$ has a state-successful query to $\Phi_i$ after having a state-successful query to $\Phi_j$.
    \item for some $i$, $\Sim$ is successful without a non-abort query to $\Phi_i$.
    \item for some $i$, $\Sim$ has a non-abort query to $\Phi_i$ where the input value on register ${\bf Z}$ is different from the output value on register ${\bf Z}$ of a non-abort query to $\Phi_{i-1}$.
\end{enumerate}
\end{claim}

\paragraph{Main Claim.} We defer the proof of \Cref{lem:structure-simulator} to \Cref{sec:proof-structure-simulator} and we finish now the proof of \Cref{lem:negative-result} by showing that ${\cal L}$ is in BQP. 

\begin{claim}
${\cal L}$ is in BQP. 
\end{claim}
\begin{proof}
We start by showing that there exists a prover $\tilde{P}$ that runs in expected polynomial time and makes the verifier $V$ accept any instance in ${\cal L}$ with probability $1-\negl(\lambda)$. 

Intuitively, $\tilde{P}$ will execute $\Sim$ and simulate the queries to $V'$ by exchanging messages with $V$. Henceforth, we will assume that the registers ${\bf X},{\bf Z}$ and ${\bf T}$ are measured in the computational basis just before making queries to $\Phi_i$: this does not affect the analysis since they are anyway measured in the channel $\Phi_i$. \\

\noindent In more detail, $\tilde{P}$ proceeds as follows: $\tilde{P}$ chooses random subspaces $S_i$ for $0 \leq i \leq k$, samples a secret-key $sk_{\text{enc}}$ corresponding to a symmetric-key encryption scheme and samples a secret-key-verification key pair $(sk_{\text{sign}},vk_{\text{sign}})$ corresponding to a digital signature scheme. 

$\tilde{P}$ then runs $\Sim$ on input $\ket{S_0}$. For every query of $\Sim$ to $\Phi_i$, $i > 1$, $\tilde{P}$ answers the query with abort. For a query made to $\Phi_1$, on input 
$\kb{\alpha_1}_{{\bf X}} \otimes {\rho_1}_{{\bf Y}} $, $\tilde{P}$ first checks if $\rho_1$ corresponds to $\ket{S_0}$ by measuring $\rho_1$ with respect to $\{\ketbra{S_0}{S_0},I-\ketbra{S_0}{S_0}\}$.
If the check fails, $\tilde{P}$ responds to the query with an abort and continues the execution of $\tilde{P}$ with the same oracle responses as before. Otherwise, $\tilde{P}$ will do the following:
\begin{enumerate}
    \item Send the message $\alpha_1$ to the verifier,
    \item Receive the message $\beta_1$ from the verifier,
    \item Compute $c_1 = \enc(sk_{\text{enc}},0)$ and $t_1 = \sign(sk_{\text{sign}},c_1)$
    \item Answer the query from $\Sim$ with $\kb{\beta_1}_{{\bf X}} \otimes \kb{S_1}_{{\bf Y}} \otimes \kb{c_1}_{{\bf Z}} \otimes \kb{t_1}_{{\bf T}}$,
    \item Set $r = 2$ and continues executing $\Sim$, simulating the queries to $\Phi_1,...,\Phi_k$ as described below.
\end{enumerate}

For every query of $\Sim$ to $\Phi_i$, for $i \ne r$, $\tilde{P}$ answers the query with abort. For a query made to $\Phi_r$, on input 
$\kb{\alpha_r}_{{\bf X}} \otimes {\rho_r}_{{\bf Y}} \otimes \kb{c'_{r}}_{{\bf Z}} \otimes \kb{t'_{r}}_{{\bf T}} $, $\tilde{P}$ first checks if $\rho_{r}$ corresponds to $\ket{S_{r-1}}$ by measuring $\rho_r$ with respect to $\{\ketbra{S_{r-1}}{S_{r-1}},I-\ketbra{S_{r-1}}{S_{r-1}}\}$. We now consider the following three cases: First, if the check fails, $\tilde{P}$ answers the query with abort and continues the execution of $\Sim$. Secondly, if the check succeeds and $c_{r-1} \ne c'_{r}$ then $\tilde{P}$ aborts. Finally, if the check succeeds and the signature $t'_r$ is a correct signature for $c'_r$ then $\tilde{P}$ will do the following: 
\begin{enumerate}
    \item Send the message $\alpha_r$ to the verifier,
    \item Receive the message $\beta_r$ from the verifier,
    \item Compute $c_r = \enc(sk_{\text{enc}},0)$ and $t_r = \sign(sk_{\text{sign}},c_r)$ (using fresh randomness),
    \item Answer the query from $\Sim$ with $\kb{\beta_r}_{{\bf X}} \otimes \kb{S_r}_{{\bf Y}} \otimes \kb{c_r}_{{\bf Z}} \otimes \kb{t_r}_{{\bf T}}$,
    \item Increment $r$ by 1, i.e., set $r = r+1$,  and continue executing $\Sim$ by answering the queries with the updated value of $r$.
\end{enumerate}

\noindent We now prove that $\tilde{P}$ makes $V$ accept any instance $x \in \lang$ with $1-\negl(\lambda)$ probability. For that we consider the following hybrids: \\

\noindent \underline{$\hybrid_0$}: The execution of $\Sim$ with access to $\Phi_1,...,\Phi_k$ that come from $V'$ defined above. \\

\noindent \underline{$\hybrid_1$}: The execution of a simulator $\Sim'$ which has the same behaviour as $\Sim$, except $\Sim'$ aborts if any of the events of the items of \Cref{lem:structure-simulator} happen. \\

From \Cref{lem:structure-simulator}, the events happen with probability at most $\negl(\lambda)$ and therefore the output of $\hybrid_0$ and $\hybrid_1$ are statistically close.\\

\noindent \underline{$\hybrid_2$}: In this hybrid, we consider the interaction of $\Sim'$ with access to $\Phi'_1,...,\Phi'_k$, which are controlled by a stateful party $P'$. At the beginning, for all $j > 1$, $P'$ answers $\Phi'_j$ with an abort. On the query to $\Phi'_1$, $P'$ will store the value of the output registers ${\bf Z}$ and ${\bf T}$, which we denote by $c_1$ and $t_1$, and replace them by $d_1 = \enc(sk_{\text{enc}},0)$ and its signature $\sign(sk_{\text{sign}},d_1)$. $P'$ will set $r = 2$ and starts answering the queries as below.

On a query to $\Phi_j$, for $j \ne r$, $P'$ answers with an abort. On a query to $\Phi_r$,  $P'$ works as follows. $P'$ will check if the content of input register ${\bf Z}$ is equal to $d_{r-1}$. If it is not, $P'$ answers with abort. Otherwise, $P'$ re-initializes the contents of the input registers  ${\bf Z}$ and ${\bf T}$ with $c_{r-1}$ and $t_{r-1}$. Then $P'$ execute the original channel of $\Phi_r$. $P'$ will store the value of the output registers ${\bf Z}$ and ${\bf T}$, which we denote by $c_r$ and $t_r$, and replace them by $d_r = \enc(sk_{\text{enc}},0)$ and its signature $\sign(sk_{\text{sign}},d_r)$. $P'$ increments $r$ and answers the queries as described above with the updated value of $r$ \\

Given the IND-CPA security of the encryption scheme, the output of $\hybrid_1$ and $\hybrid_2$ are computationally close. \\

\noindent \underline{$\hybrid_3$}: We consider the interaction of $\Sim'$ with the prover $\tilde{P}$ of the protocol. \\

The output of the simulator in $\hybrid_2$ and $\hybrid_3$ are the same, since $\Sim'$ does not have access to any value that depends on the internal state of $V$ in $\hybrid_2$, and $\tilde{P}$ simulates all operations of $V'$.

\medskip

We show now that the successful output of the simulator in $\hybrid_3$ implies that $V$ accepts in the protocol.  The key observation is that, as described above, if the simulator does not abort, any state-successful query to $\Phi_k$ is also non-abort. But by construction, this is only possible if $V$ accepts in the protocol. Therefore, since $\Sim$ does not abort with probability at least $1-\negl(\lambda)$, $V$ accepts with the same probability when interacting with $\tilde{P}$ on a yes-instance. 

We also notice that for any no-instance $x$, by the soundness of the protocol, we have that such a $\tilde{P}$ makes $V$ accept with probability at most $\negl(\lambda)$.

Finally, we need to discuss the runtime of $\tilde{P}$. So far, we have that there is a polynomial $q$ such that the expected runtime of $\tilde{P}$ if $q(\lambda)$. However, we can consider a prover $\tilde{P}'$ that terminates the execution of $\tilde{P}$ after $q(\lambda)^2$ steps and in this case, the probability that $\tilde{P}'$ convinces the honest verifier is still a constant.. In this case, both $\tilde{P}'$ and $V$ runs in polynomial-time, which can be simulated by a BQP algorithm, which finishes the proof.
\end{proof}
\end{proof}

\subsection{Proof of \Cref{lem:structure-simulator}}
\label{sec:proof-structure-simulator}

We prove each item of \Cref{lem:structure-simulator} in a separate lemma below.

\begin{lemma}\label{lem:skip_query}
Except with $ \negl(\lambda)$ probability, for no $i$, $\Sim$ has a state-successful query to $\Phi_{i}$ without having a non-abort query to $\Phi_{i-1}$;
\end{lemma}
\begin{proof}
We start by proving this claim assuming that $\Sim$ runs in strict polynomial time, and then we generalize to the case where $\Sim$ runs in expected polynomial time.

Let $J$ be the event where $\Sim$ has a state-successful query to $\Phi_{i}$, for some $i$, without having a non-abort query to $\Phi_{i-1}$.\footnote{As defined in the proof of \Cref{lem:negative-result}, a query is state-successful if it passes the subspace state check of $\Phi_i$.}

Let us define $J_{i}$ the event where $\Sim$ has a successful query to $\Phi_{i}$ without having a non-abort query to $\Phi_{i-1}$. Notice that, $J = \bigvee_{0 \leq i \leq k} J_{i}$, and by union bound
\begin{align}\label{eq:bound_query_skip_total}
\Pr[J] \leq \sum_{0 \leq i \leq k} \Pr[J_{i}],
\end{align}
and we follow by showing that for any $i$
\begin{align}\label{eq:bound_query_skip}
     \Pr[J_{i}] = \negl(\lambda),
\end{align}
which implies \Cref{eq:bound_query_skip_total} and the proofs this lemma.

  We prove \Cref{eq:bound_query_skip} by contradiction, so let us assume that there exists some polynomial $p$ such that $\Pr[J_{i}] > 1/p(\lambda)$. We
  show that we can construct an adversary $\calA_{i}$ that is able to clone a random subspace state $\ket{A}$ with polynomially many queries to the verification oracle, contradicting \Cref{lem:money}.

  $\calA_{i}$ works as follows: $\calA_i$ receives some state $\ket{A}$ and
  oracle access to $U_A$ (as described in~\Cref{sec:querylb}). $\calA_i$  keeps $\ket{A}$ aside and simulates $\Sim$ with oracle access to the channels
  $\Phi'_1,...,\Phi'_{k}$ that we define now. For $j \not\in \{i-1,i\}$, we have  $\Phi'_j = \Phi_j$.
  Whenever $\Sim$ queries $\Phi'_{i-1}$, $\calA_{i}$ runs $\Phi_{i-1}$ and if
  the answer is non-abort, then $\calA_i$ aborts.  For $\Phi'_{i}$,
  $\calA_{i}$ uses the oracle $U_A$ to verify the subspace state and if it passes, then
  $\calA_{i}$ outputs the result of the query and the original state $\ket{A}$. Otherwise
  $\calA_i$ answers the query to $\Sim$ with an abort. If $\Sim$
  terminates before a state-successful query to $\Phi_i'$, $\calA_{i}$ aborts. 

  We argue now that $\Pr[J_{i}] > 1/p(\lambda)$ implies that   $\calA_i$ is able to output two
  copies of $\ket{A}$ with probability at least $1/p(\lambda)$ with polynomially-many oracle calls to $U_A$. Notice that whenever
  $\calA_i$ does not abort, it outputs $\ket{A}^{\otimes 2}$ as desired: one
  of the copies is guaranteed to be $\ket{A}$ since it was received by the
  challenger in the quantum money scheme whereas the second copy passed
  verification test using $U_A$, and thus by \Cref{lem:uncstates:test} it must be on state $\ket{A}$. In this
  case, we need to prove the probability that $\calA_i$ does not abort. 

  Let us define the following two events:
  \begin{enumerate}
    \item[$E_1$]: $\Sim$ terminates before a successful query to $\Phi_i$;
    \item[$E_2$]: $\Sim$ queries $\Phi_{i-1}$ with non-abort answer before querying
      $\Phi_{i}$   
  \end{enumerate}

  Let also $E$ be the event on which $\calA_i$ aborts. By construction we have
  that $E = E_1 \vee E_2$.  Moreover, we have that $E_j \subseteq
  \overline{J_i}$ for $j \in \{1,2\}$, and therefore
  $$\Pr[E] \leq \Pr[\overline{J_i}] = 1 - \Pr[J_i] \leq 1 - 1/p(\lambda).$$

  It follows then that the probability that $\calA_i$ does not abort (and
  outputs $\ket{A}^{\otimes 2}$) is at least $1/p(\lambda)$ with polynomially-many queries to $U_A$, which is a contradiction to \Cref{lem:money} and this finishes the proof of \Cref{eq:bound_query_skip}.

\medskip

  Finally, to finish the proof, we need to consider the case the expected runtime of $\Sim$ $q(\lambda)$, for some polynomial $q$. Let us suppose that there exists some non-zero polynomial $p$ such that $\Pr[J_i] \geq 1/p(\lambda)$.

In this case, we can consider a simulator $\Sim'$ that executes $\Sim$ but aborts if the runtime of $\Sim$ is greater than $ q(\lambda)^ p(\lambda)$. Notice that $\Sim'$ aborts with probability at most $\frac{1}{q(\lambda)p(\lambda)}$ and we have that $\Pr[J_{i}] > 1/p(\lambda)$ holds for $\Sim'$, such that $\Pr[J_i] \geq 1/p(\lambda) - \frac{1}{q(\lambda)p(\lambda)} = \frac{1}{\poly(\lambda)}$, which is a contradiction.
\end{proof}

\begin{lemma}\label{lem:two-queries}
Except with $\negl(\lambda)$ probability, for no $i \leq j$, $\Sim$ has a state-successful query to $\Phi_i$ after having a state-successful query to $\Phi_j$.
\end{lemma}
\begin{proof}
As in \Cref{lem:skip_query}, we prove the lemma for strict polynomial-time $\Sim$ and the extension to expected polynomial-time $\Sim$ follows as before.

We first prove it for $i = j$ and we then generalize it for $i \leq j$. 
Let $B_{i}$ be the event where $\Sim$ has two state-successful queries to $\Phi_i$. 
Our goal now is to show that for every polynomial $p$,
\begin{align}
\label{eq:money-case2}
    \Pr[B_{i}] < i \cdot \frac{1}{p(\lambda)},
\end{align}
which implies \Cref{lem:two-queries} by a union bound on $\Pr\left[\bigvee_{1 \leq i  \leq k} B_i\right]$. We prove now 
\Cref{eq:money-case2} by induction on $i$.

Let us start with the base case of \Cref{eq:money-case2}, where $i = 1$.
Our proof follows by contradiction. Therefore, let us suppose that  there exists some polynomial $p$  such that $\Pr[B_{1}] \geq 1/p(\lambda)$.  We show then that we can construct an adversary $\calA_{1}$ that is able to clone a random subspace state $\ket{A}$ with polynomially many queries to the verification oracle, contradicting \Cref{lem:money}.

$\calA_{1}$ receives some state $\ket{A}$ and oracle access to $U_A$, and it proceeds as follows: 
$\calA_{1}$ sets $q = 0$ and simulates
$\Sim$ with side information $\ket{A}$ and oracle access to $\Phi'_1,...,\Phi'_{k}$, where $\Phi'_j = \Phi_j$ for $j > 1$.  For $\Phi'_1$, $\calA_1$ uses
the oracle $U_A$ to verify the quantum state. If the test does not pass, $\calA_1$ answers the query with abort. If the test passes and $q = 0$, $\calA_1$ stores the output of $U_A$ on some register ${\bf M}$, sets $q = 1$ and 
continues the simulation of $\Sim$. If the test passes and $q = 1$, $\calA_1$ outputs the state returned from the second query to $U_A$, along with the state stored in register ${\bf M}$.
If $\Sim$ terminates, then $\calA_1$ aborts.

  We argue now that if $\Pr[B_1] \geq 1/p(\lambda)$, then $\calA_i$ is able to output two
  copies of $\ket{A}$ with probability at least $1/p(\lambda)$ with polynomially many queries to $U_A$. Notice that whenever
  $\calA_i$ does not abort, it outputs $\ket{A}^{\otimes 2}$ as desired, since both 
  both of the states passed the
  verification test using $U_A$ (by \Cref{lem:uncstates:test}). In this
  case, we need to bound the probability that $\calA_i$ does not abort. 

  Let $F$ be the event when $\calA_1$ aborts, which happens if $\Sim$ terminates without two state-successful queries to $\Phi'_i$.
  
  Notice that $F \subseteq
  \overline{B_1}$, and therefore
  $$\Pr[F] \leq \Pr[\overline{B_1}] = 1 - \Pr[B_1] \leq 1 - 1/p(\lambda).$$

  It follows then that the probability that $\calA_1$ does not abort (and
  outputs $\ket{A}^{\otimes 2}$) is at least $1/p(\lambda)$, which contradicts \Cref{lem:money} and therefore we have \Cref{eq:money-case2} for $i=1$.

  \medskip

  We assume now that \Cref{eq:money-case2} works for $i-1$, and we prove it for $i$. The proof is similar to the base case, but we present here for completeness. We prove it again by contradiction. Therefore, let us suppose that $    \Pr[B_{i}] \geq i \cdot p(\lambda)$, for some polynomial $p$. Notice that
  \begin{align*}
  i \cdot p(\lambda) &\leq \Pr[B_{i}] \\
  &= \Pr[B_{i} \wedge \overline{B_{i-1}}]
  + \Pr[B_{i}  | B_{i-1}]\Pr[B_{i-1}] \\
  &< \Pr[B_{i} \wedge \overline{B_{i-1}}]
  + (i-1) \cdot p(\lambda),
  \end{align*}
  where the inequality holds by the induction hypothesis. Therefore, our assumption implies that
\begin{align}\label{eq:money-case3}
    \Pr[B_{i} \wedge \overline{B_{i-1}}] \geq p(\lambda).
\end{align}

  We show then that if \Cref{eq:money-case3} holds, then we can construct an adversary $\calA_{i}$ that is able to clone a random subspace state $\ket{A}$ with polynomially many queries to the verification oracle, contradicting \Cref{lem:money}.

$\calA_{i}$ receives the state $\ket{A}$ and oracle access to $U_A$, and it proceeds as follows:
$\calA_{i}$ sets $q_{i-1} = q_i = 0$, keeps $\ket{A}$ on a register ${\bf N}$ aside and simulates
$\Sim$ on $\Phi'_1,...,\Phi'_{k}$, where $\Phi'_j = \Phi_j$ for $j \not\in
\{i-1,i\}$. For $\Phi'_{i-1}$, if $q_{i-1} = 1$, $\calA_{i}$ aborts. Otherwise,
$\calA_{i}$ simulates $\Phi_{i-1}$, and if it does not abort, $\calA_{i}$ sets
$q_{i-1} = 1$ and stores the quantum state $\ket{S_{i-1}}$ returned by $\Phi_{i-1}$ and replaces it by the given copy of $\ket{A}$. For $\Phi'_{i}$, $\calA_{i}$ uses
the oracle to verify the quantum state. If it does not pass the test,
$\calA_i$ answers the query to $\Phi'_{i}$ with abort.
If the test passes, $\calA_i$ has a different behaviour depending on the value
of $q_i$
\begin{itemize}
  \item If the test passes and $q_i = 0$, then $\calA_{i}$ puts the state returned by the oracle $U_A$ on a register ${\bf M}$, replaces it by $\ket{S_{i-1}}$, simulates $\Phi_i$ and sets $q_i = 1$.
  \item If the test passes and $q_i = 1$, $\calA_i$ outputs the state returned by the oracle $U_A$ along with the state in register ${\bf M}$.
\end{itemize}
If $\Sim$ terminates, then $\calA_i$ aborts.

  We argue now that \Cref{eq:money-case3} implies that $\calA_i$ is able to output two
  copies of $\ket{A}$ with probability at least $1/p(\lambda)$ with polynomially many queries. Notice that whenever
  $\calA_i$ does not abort, it outputs $\ket{A}^{\otimes 2}$ as desired, since
  both of the states passed the
  verification test using $U_A$ (by \Cref{lem:uncstates:test}). In this
  case, we need to bound the probability that $\calA_i$ does not abort. 

  Let us define the following events:
  \begin{enumerate}
    \item[$F_1$]: $\Sim$ aborts due to two non-abort queries to $\Phi'_{i-1}$;
    \item[$F_2$]: $\Sim$ terminates without two queries to $\Phi'_i$ that pass
      the verification of the subset state;
  \end{enumerate}

  Let also $F$ be the event on which $\calA_i$ aborts. By construction we have
  that $F = F_1 \vee F_2$.  Notice that $F_1 \subseteq B_{i-1}$ and that $F_2 \subseteq
  \overline{B_i}$. It follows from \Cref{eq:money-case3} that
  $$\Pr[F] \leq \Pr[B_{i-1} \vee \overline{B_i} ] = 1 - \Pr[\overline{B_{i-1}} \wedge B_i ] \leq 1 - p(\lambda).$$

  Therefore, the probability $\calA_i$ does not abort (and
  outputs $\ket{A}^{\otimes 2}$) is at least $p(\lambda)$, which contradicts \Cref{lem:money} and we have \Cref{eq:money-case2}. 
  
 \medskip
  
  To finish the proof, we have to prove the lemma for $i < j$. Let $C_{i,j}$ be the event where $\Sim$ has a state-successful query to $\Phi_i$ after having a state-successful query to $\Phi_j$. By defining $J$ to be the event where $\Sim$ has a state-successful query to $\Phi_{i}$, for some $i$, without having a non-abort query to $\Phi_{i-1}$, we have that
  $$ \Pr[C_{i,j}] = \Pr[C_{i,j}|J]\Pr[J] +
\Pr[C_{i,j} \wedge \overline{J}] \leq \negl(\lambda) + \Pr[B_{j}] = \negl(\lambda),
$$
where the first inequality follows from \Cref{lem:skip_query} and $B_j \subseteq C_{i,j} \wedge \overline{J}$, since $\overline{J}$ implies state-successful query to $\Phi_j$ before a state-successful query to $\Phi_i$ and $C_{i,j}$ implies a state-successful query to $\Phi_j$ after a state-successful query $\Phi_i$.
\end{proof}

\begin{lemma}
Except with $\negl(\lambda)$ probability, for no $i$, $\Sim$ is successful without a non-abort query to $\Phi_i$.
\end{lemma}
\begin{proof}
As in \Cref{lem:skip_query}, we prove the lemma for strict polynomial-time $\Sim$ and the extension to expected polynomial-time $\Sim$ follows as before. Notice that we need  prove it for $i = k$, since it will imply, along with \Cref{lem:skip_query}, the statement for every $1 \leq i < k$.

Let $E$ be the event where $\Sim$ is successful and does not have a non-abort query to $\Phi_k$. We show that if there exists a polynomial $p$ such that
\begin{align}\label{eq:bound-success}
    \Pr[E] > \frac{1}{p(\lambda)},
\end{align}
then we can construct an adversary $\calA$ that is able to clone a random subspace state $\ket{A}$ with polynomially-many queries to the verification oracle $U_A$, contradicting \Cref{lem:money}.  

For that, we will use the zero-knowledge property of the protocol that implies that for every distinguisher, the output of $\Sim^{\channel(V')}$ is indistinguishable from the output of $V'$ while interacting with the honest prover. In particular, we will consider $V'$ where $S_{i+1} = A$ and distinguisher $\mathcal{D}$, using an oracle to $U_A$, checks if the output register ${\bf Y}$  of $\Sim$ is indeed $\ket{A}$ (which is always the case for $V'$). 

$\calA$ receives some state $\ket{A}$, and it proceeds as follows: 
$\calA$  puts $\ket{A}$ on register ${\bf M}$ aside, and simulates $\Sim$ with oracle access to the channels
  $\Phi'_1,...,\Phi'_{k}$, where $\Phi'_j = \Phi_j$ for $j \ne k$.
  Whenever $\Sim$ queries $\Phi'_{k}$, $\calA$ uses the oracle simulates $\Phi_k$ up to step $6$ (where $V'$ runs the original verifier $V$). If $\Phi_k$ did not abort up to this point, then $\calA$ aborts. Otherwise, $\calA$ answer this query with abort and continues the execution of $\Sim$. If $\Sim$
  aborts, then $\calA$ also aborts. Finally, if $\Sim$ terminates successfully, $\calA$ outputs the register ${\sf Y}$ of $\Sim$'s output, along with the state $\ket{A}$ stored in ${\bf M}$. 
  
  Let $F$ be the event where the output register ${\bf Y}$ of $\Sim$ projects into $\ket{A}$. In order to $\calA$ to output two copies of $\ket{A}$ we need that $E \wedge F$ hold.
  
  By the zero-knowledge property, we have that 
      $\Pr[F] > 1 - \negl(\lambda)$. This fact, along with \Cref{eq:bound-success}, gives us that
$$ \Pr[E \wedge F] \geq \Pr[E] + \Pr[F] -1 \geq \frac{1}{p(\lambda)} - \negl(\lambda) = \frac{1}{\poly(\lambda)}.$$

Such $\calA$ contradicts \Cref{lem:money}, and therefore \Cref{eq:bound-success} is false.
\end{proof}

\begin{lemma}
Except with $\negl(\lambda)$ probability, for no $i$, for some $i$, $\Sim$ has a non-abort query to $\Phi_i$ where the input value on register ${\bf Z}$ is different to the output value on register ${\bf Z}$ of a non-abort query to $\Phi_{i-1}$.
\end{lemma}
\begin{proof}
As in \Cref{lem:skip_query}, we prove the lemma for strict polynomial-time $\Sim$ and the extension to expected polynomial-time $\Sim$ follows as before.

Let $D$ be the event where there exists some $i$ such that $\Sim$ has a non-abort query to $\Phi_i$ where the input value on register ${\bf Z}$ is different to the output value on register ${\bf Z}$ of a non-abort query to $\Phi_{i-1}$.

We show that if there is a polynomial $p$ such that
\begin{align}\label{eq:bound-forgery}
Pr[D] > \frac{1}{p(\lambda)}
\end{align}
then we can construct an adversary $\calA$ that breaks the digital signature scheme.

$\calA$ receives a verification key $vk$ and has oracle access to the $\sign_{sk}$. $\calA$ initializes an empty list $L$ and proceeds as follows: 
$\calA$  simulates
$\Sim$ on $\Phi'_1,...,\Phi'_{k}$, where the query to $\Phi'_i$ behaves as follows. If $i = 1$, $\calA$ simulates the behaviour of $\Phi_1$. Let $c$ be the content of the output register ${\bf Z}$. $\calA$ adds $c$ to $L$ and reinitialize the register ${\bf T}$ with $\sign_{sk}(c)$.

On query to $\Phi'_i$ for $i > 1$, let $m$ be the content of the input register ${\bf Z}$ and $s$ be the content of the input register ${\bf T}$. $\calA$ checks if $s$ is a valid signature of $m$ and if $m$ does {\em not} appear in $L$. If the check passes, then $\calA$ outputs $(m,s)$. Otherwise, $\calA$ continues the simulation of $\Phi_i$. 
Let $c$ be the content of the output register ${\bf Z}$. $\calA$ adds $c$ to $L$ and reinitialize the register ${\bf T}$ with $\sign_{sk}(c)$.

If $\Sim$ terminates, then $\calA$ aborts.

We notice that whenever $\calA$ does not abort, it outputs a pair of message and signature $(m,s)$ such that $m$ was not and input to a signing query, and therefore $\calA$ is able to forge a signature for $m$.

To finish the proof, we need to bound the abort probability of $\calA$. We notice that $D$ is exactly the case when  $\calA$ does not abort, and by \Cref{eq:bound-forgery}, this happens with probability at least $\frac{1}{p(n)}$. 
\end{proof}

\printbibliography

\end{document}